\documentclass[envcountsame,runningheads]{llncs} 

\usepackage[noadjust]{cite}
\usepackage[cmex10]{amsmath}
\usepackage{amssymb,stmaryrd,mathtools,xspace,wrapfig}

\usepackage{microtype}

\usepackage{concurmacros}
\usepackage{tikz}
\usetikzlibrary{positioning,shapes,automata,fit,arrows}

\tikzset{may/.style={dashed,->,auto},
    must/.style={->,auto},
    multimust/.style={-},
    process/.style={inner sep=1.2pt,outer sep=3pt,circle,fill,draw,on grid},
    fork/.style={inner sep=0pt,outer sep=0pt,%circle,draw,
		on grid}
}

\newcommand{\edgelabel}[1]{\mathsf{#1}}
\newcommand{\req}{\edgelabel{request}}
\newcommand{\grant}{\edgelabel{grant}}
\newcommand{\idle}{\edgelabel{idle}}
\newcommand{\work}{\edgelabel{work}}

\newcommand{\ff}{\mathbf{ff}}
\newcommand*\ttt{\mathbf{t\!t}}
\newcommand{\by}{/}

\newcommand{\rhml}{$\nu$HML\xspace}

\newcommand{\Pfin}[1]{\ensuremath{2_{ \textup{Fin}}^{#1}}}

\newcommand{\figspaceb}{\vspace*{-3em}}
\newcommand{\figspacee}{\vspace*{-2em}}
\newcommand{\myspace}{\vspace*{-1em}}

\title{\hspace*{-1ex}\mbox{Hennessy-Milner Logic with Greatest Fixed
    Points} as a Complete Behavioural Specification Theory}
\titlerunning{Hennessy-Milner Logic with Greatest Fixed Points}

\author{%
     Nikola Bene{\v s}\inst 1\thanks{The author has been supported by the Czech Science Foundation grant No. GAP202/11/0312.}
    \and
    Beno\^{\i}t Delahaye\inst 2 \and
    Uli Fahrenberg\inst 2 \and
    \\Jan K{\v r}et{\'\i}nsk{\'y}\inst{1,3}\thanks{The author is partially supported by the Czech Science Foundation, project No. P202/10/1469.} 
    \and Axel Legay\inst 2
  }

\institute{ Masaryk University, Brno, Czech Republic \and
  Irisa / INRIA Rennes, France \and
 Technische Universit{\"a}t M{\"u}nchen, Germany}

\authorrunning{Bene{\v s}, Delahaye, Fahrenberg, K{\v
    r}et{\'\i}nsk{\'y}, Legay}

\begin{document} 

\maketitle 
\vspace*{-1.5em}

\begin{abstract}
  There are two fundamentally different approaches to specifying and
  verifying properties of systems. The \emph{logical} approach makes use
  of specifications given as formulae of temporal or modal logics and
  relies on efficient model checking algorithms; the \emph{behavioural}
  approach exploits various equivalence or refinement checking methods,
  provided the specifications are given in the same formalism as
  implementations.

  In this paper we provide translations between the logical formalism of
  Hennessy-Milner logic with greatest fixed points and the behavioural
  formalism of disjunctive modal transition systems.
  We also introduce a new operation of quotient for the above equivalent
  formalisms, which is adjoint to structural composition and allows
  synthesis of missing specifications from partial implementations.
  This is a~substantial generalisation of the quotient for deterministic
  modal transition systems defined in earlier papers.
\end{abstract}

\vspace*{-2.5em}

\section{Introduction}\label{sec:intro}

There are two fundamentally different approaches to specifying and
verifying properties of systems. Firstly, the \emph{logical} approach
makes use of specifications given as formulae of temporal or modal
logics and relies on efficient model checking algorithms. Secondly, the
\emph{behavioural} approach exploits various equivalence or refinement
checking methods, provided the specifications are given in the same
formalism as implementations.

In this paper, we discuss different formalisms and their relationship.
As an example, let us consider labelled transition systems and the
property that \emph{``at all time points after executing $\req$, no
  $\idle$ nor further requests but only $\work$ is allowed until
  $\grant$ is executed''}. The property can be written in
\eg~CTL~\cite{DBLP:conf/lop/ClarkeE81} as
\begin{equation*}
  \text{AG}(\req\Rightarrow \text{AX}(\work \text{ AW } \grant))
\end{equation*}
% or in modal \emph{$\mu$-calculus} as $$...awful...$$
or as a recursive system of equations in Hennessy-Milner
logic~\cite{DBLP:journals/tcs/Larsen90} as
\begin{align*}
X &= [\grant,\idle,\work]X \wedge [\req] Y\\
Y &= (\langle\work\rangle Y\vee \langle\grant\rangle X)\wedge[\idle,\req]\ff
\end{align*}
where the solution is given by the greatest fixed point.

As formulae of modal logics can be difficult to read, some people prefer
automata-based behavioural specifications to logical ones.  One such
behavioural specification formalism is the one of disjunctive modal
transition systems (DMTS) \cite{DBLP:conf/lics/LarsenX90}.
Fig.~\ref{fig:intrexample} (left) displays a~specification of our
example property as a~DMTS.  Here the dashed arrows indicate that the
transitions \emph{may or may not} be present, while branching of the
solid arrow indicates that at least one of the branches \emph{must} be
present.  An example of a labelled transition system that
\emph{satisfies} our logical specifications and \emph{implements} the
behavioural one is also given in Fig.~\ref{fig:intrexample}.
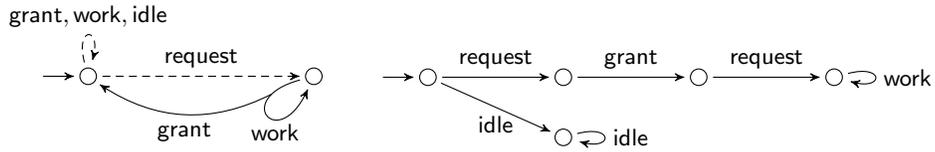
\begin{figure}[t]
%\figspacee
  \centering
  \begin{tikzpicture}[x=3cm,y=2cm,font=\footnotesize,->,>=stealth',
    state/.style={shape=circle,draw,font=\scriptsize,inner sep=.5mm,outer
      sep=0.8mm,minimum size=0.22cm,initial text=,initial
      distance=3ex}]
    \node[state,initial] (X) {};
    % \node (XX) at(-0.35,0){};
    % \path (XX) edge [->] (X.west);
    \node[state] (Y) at (1,0) {};
%%    \coordinate (Z) at (.3,-.4) {};
    \path[densely dashed,->] (X) edge node[above]{$\req$} (Y); 
    \path[densely dashed,->] (X) edge[loop above]
    node[above]{$\grant,\work,\idle$} (X);
%%    \path[-,thick] (Y)  edge (Z); 
%%    \path[->] (Z) edge[bend right] node[right]{$\work$} (Y); 
%%    \path[->] (Z) edge node[left,pos=.4]{$\grant\,\,$} (X); 
    \path (Y) edge[loop,in=-115,out=-165,looseness=15] node[inner sep=0,outer sep=0,minimum size=0,pos=0.2,name=YY] {} node[below] {$\ \ \work$} (Y);
    \path (YY) edge[bend left] node[below] {$\grant$} (X);
  \end{tikzpicture}~~~~~~~
\begin{tikzpicture}[x=1.8cm,y=2cm,font=\footnotesize,->,>=stealth',
    state/.style={shape=circle,draw,font=\scriptsize,inner
      sep=.5mm,outer sep=0.8mm, minimum size=0.22cm,initial
      text=,initial distance=3ex}]
    \node[state, initial left] (X) {};
    \node[state] (Y) at (1,0) {};
    \node[state] (Z) at (2,0) {};
    \node[state] (V) at (3,0) {};
    \node[state] (W) at (1,-.4) {};
    \path[->] (X) edge node[above]{$\req$} (Y); 
    \path[->] (Y) edge node[above]{$\grant$} (Z); 
    \path[->] (Z) edge node[above]{$\req$} (V); 
    \path[->] (V) edge[loop right] node[right]{$\work$} (V); 
    \path (X)  edge node[below]{$\idle$} (W); 
    \path[->] (W) edge[loop right] node[right]{$\idle$} (W); 
  \end{tikzpicture}
  \caption{%
    \label{fig:intrexample}
    DMTS specification corresponding to 
    $\text{AG}(\req\Rightarrow \text{AX}(\work \text{ AW } \grant))$, 
and its implementation
    %with action alphabet $\{ \grant, \work, \idle, \req\}$
}
\figspacee
\end{figure}

The alternative between logical and behavioural specifications is not
only a~question of preference.  Logical specification formalisms put a
powerful logical language at the disposal of the user, and the logical
approach to model
checking~\cite{DBLP:conf/programm/QueilleS82,DBLP:conf/lop/ClarkeE81}
has seen a lot of success and tool implementations.  Automata-based
specifications~\cite{DBLP:conf/avmfss/Larsen89,DBLP:conf/cav/BrunsG99},
on the other hand, have a focus on \emph{compositional} and
\emph{incremental} design in which logical specifications are somewhat
lacking, with the trade-off of generally being less expressive than
logics.

To be more precise, automata-based specifications are, by design,
compositional in the sense that they support structural
\emph{composition} of specifications and, in most cases, its adjoint,
\emph{quotient}.  This is useful, even necessary, in practical
verification, as it means that (1) it is possible to infer properties of
a system from the specifications of its components, and (2) the problem
of correctness for a system can be decomposed into verification problems
for its components.  We refer to~\cite{DBLP:conf/concur/Larsen90} for a
detailed account on composition and decomposition.

It is thus desirable to be able to translate specifications from the
logical realm into behavioural formalisms, and \textit{vice versa} from
behavioural formalisms to logic-based specifications.  This is, then,
the first contribution of this paper: we show that Hennessy-Milner logic
with greatest fixed points (\rhml) and DMTS (with several initial states)
are equally expressive, and we provide translations forth and back.  For
doing this, we introduce an~auxiliary intermediate formalism NAA
(a nondeterministic extension of acceptance
automata~\cite{DBLP:journals/jacm/Hennessy85,DBLP:journals/entcs/Raclet08})
% new formalism of \emph{Boolean formulae with
%  states} (BFS), which in a sense belongs to both the logical and the
%behavioural world and 
which is equivalent in expressiveness to both \rhml and DMTS.

% Additionally, we show how to introduce structural composition and
% quotient for NAA.  
% Using the
% translations, one gets a notion of structural composition for \rhml
% which satisfies the important property that $M_1\models \phi_1$ and
% $M_2\models \phi_2$ imply $M_1\| M_2\models \phi_1\| \phi_2$, which
% permits \emph{compositional reasoning} in model checking.
We also discuss other desirable features of specification formalisms,
namely structural composition and quotient.  As an example, consider a
specification~$S$ of the final system to be constructed and $T$ either
an already implemented component or a specification of a service to be
used. The task is to construct the most general specification of the
rest of the system to be implemented, in such a~way that when composed
with any implementation of $T$, it conforms with the specification $S$.
This specification is exactly the quotient $S\by T$.

\subsubsection{Contribution}

Firstly, we show that the formalisms of \rhml, NAA and DMTS have the
same expressive power, and provide the respective translations. As 
a~result, the established connection allows for a graphical representation
of \rhml as DMTS. This extends the graphical representability of HML
without fixed points as modal transition
systems~\cite{DBLP:conf/avmfss/Larsen89,DBLP:journals/tcs/BoudolL92}.
In some sense this is optimal, as due to the alternation of least and
greatest fixed points, there seems to be no hope that the whole
$\mu$-calculus could be drawn in a similarly simple way.

Secondly, we show that there are natural operations of conjunction and
disjunction for NAA which mimic the ones of \rhml.  As we work with
multiple initial states, disjunction is readily defined, and conjunction
extends the one for DMTS~\cite{DBLP:conf/atva/BenesCK11}.
% Thus, NAA form a bounded distributive lattice up to equivalence.
%
Thirdly, we introduce structural composition on NAA. For simplicity we assume
CSP-style synchronisation of labels, but the construction can easily
be generalised to other types of label synchronisation. 

Finally, we provide a solution to the open problem of the general
quotient.  We extend the quotient constructions for deterministic modal
transition systems (MTS) and acceptance
automata~\cite{DBLP:journals/entcs/Raclet08} to define the quotient for
the full class of (possibly nondeterministic) NAA. We also provide a
more efficient procedure for (possibly nondeterministic) MTS. These
constructions are the technically most demanding parts of the paper.

With the operations of structural composition and quotient, NAA, and
hence also DMTS and \rhml, are fully compositional behavioural
specification theories and form a \emph{commutative residuated
  lattice}~\cite{journal/ijac/HartRT02,journal/trams/WardD39} up to
equivalence. This makes a rich algebraic theory available for
compositional reasoning about specifications.
%
% In other words, a quotient is a compact representation of all the
% implementations that satisfy $S$ when put in parallel with $T$.  The
% best one can then be chosen according to e.g.~some non-functional
% criteria.
%
% For simplicity, consider full synchronization $\|$. Now given
% specifications $S$ and $T$ the task is to find the most permissive
% specification $S\by T$ such that $S$ is still at least as permissive
% as $T\ \|\ S\by T$. This question naturally arises in
% component-based design.
%
Most of the constructions we introduce are implemented in a prototype tool~\cite{BMoTraS}.
% accessible at http://delahaye.benoit.free.fr/BMoTraS.tar
%
Due to space constraints, some of the proofs
% can be found in the appendix.
had to be omitted from the paper.

\subsubsection{Related work}

Hennessy-Milner logic with
recursion~\cite{DBLP:journals/tcs/Larsen90} %,ReactiveSystems}
is a popular logical specification formalism which 
%is expressionally equivalent to
has the same expressive power as
%NO! CTL~\cite{DBLP:conf/lop/ClarkeE81} and
$\mu$-calculus~\cite{DBLP:journals/tcs/Kozen83}.  It is obtained from
Hennessy-Milner logic (HML)~\cite{DBLP:journals/jacm/HennessyM85} by
introducing variables and greatest and least fixed points.
Hennessy-Milner logic with \emph{greatest} fixed points (\rhml) is
equivalent to $\nu$-calculus, \ie~$\mu$-calculus with greatest fixed
points only.
%, and with the safety fragment of CTL.\todo{J:is it? mu-calculus is not CTL!}

DMTS have been proposed as solutions to algebraic process equations
in~\cite{DBLP:conf/lics/LarsenX90} and further investigated also as a
specification
formalism~\cite{DBLP:conf/concur/Larsen90,DBLP:conf/atva/BenesCK11}. The
DMTS formalism is a member of the modal transition systems (MTS) family
and as such has also received attention recently. The MTS formalisms
have proven to be useful in practice. Industrial applications started as
early as~\cite{DBLP:journals/scp/Bruns97} where MTS have been used for
an air-traffic system at Heathrow airport. Besides, MTS classes are
advocated as an appropriate base for interface theories
in~\cite{RB-acsd09} and for product line theories
in~\cite{nyman2008modal}. Further, an MTS based software engineering
methodology for design via merging partial descriptions of behaviour has
been established in~\cite{DBLP:conf/sigsoft/UchitelC04} and methods for
supervisory control of MTS shown in~\cite{Darondeau2010a}.  Tool support
is quite extensive,
\eg~\cite{DBLP:journals/fmsd/BorjessonLS95,DBLP:conf/eclipse/DIppolitoFFU07,DBLP:conf/atva/BauerML11,DBLP:conf/atva/BenesCK11}.

Over the years, many extensions of MTS have been proposed. While MTS can
only specify whether or not a particular transition is required, some
extensions equip MTS with more general abilities to describe what
\emph{combinations} of transitions are possible. These include
DMTS~\cite{DBLP:conf/lics/LarsenX90},
1-MTS~\cite{DBLP:journals/jlp/FecherS08} allowing to express exclusive
disjunction, OTS~\cite{DBLP:conf/memics/BenesK10} capable of expressing
positive Boolean combinations, and Boolean
MTS~\cite{DBLP:conf/atva/BenesKLMS11} covering all Boolean
combinations. The last one is closely related to our NAA, the acceptance
automata
of~\cite{DBLP:journals/jacm/Hennessy85,DBLP:journals/entcs/Raclet08}, as
well as hybrid modal logic~\cite{Prior68,Blackburn00}.
%
% Further, MTS
% framework has also been lifted to \emph{quantitative
%   settings}~\cite{MFCS11,DBLP:conf/csr/BauerFLT12,DBLP:journals/mscs/BauerJLLS12}.
% This includes probabilistic~\cite{TCS11,VMCAI11} and
% %even more importantly 
% timed systems~\cite{DBLP:conf/cav/CeransGL93,
% %JLS:weighted,%
%   DBLP:conf/lpar/BenesKLMS12,DBLP:conf/atva/DavidLLNW10,timedmodal}
% with clear applications in embedded systems design.

Larsen has shown in~\cite{DBLP:conf/avmfss/Larsen89} that any finite MTS
is equivalent to a HML formula (without recursion or fixed points), the
\emph{characteristic formula} of the given MTS.  Conversely, Boudol and
Larsen show in~\cite{DBLP:journals/tcs/BoudolL92} that any consistent
and \emph{prime} HML formula is equivalent to a~MTS.
% Larsen's~\cite{DBLP:conf/concur/Larsen90} contains more thoughts and
% results on composition and decomposition in specification theories.
%
% (a formula is prime if implying a disjunction means implying one of
% the alternatives).
Here we extend these results to \rhml formulae, and show that any such
formula is equivalent to a DMTS, solving a problem left open
in~\cite{DBLP:conf/lics/LarsenX90}.  Hence \rhml supports full
compositionality and decomposition in the sense
of~\cite{DBLP:conf/concur/Larsen90}.  This finishes some of the work
started
in~\cite{DBLP:conf/avmfss/Larsen89,DBLP:journals/tcs/BoudolL92,DBLP:conf/concur/Larsen90}.
% Recently, the graphical representability of a variant of alternating
% simulation called covariant-contravariant simulation has been studied
% in~\cite{DBLP:journals/corr/abs-1108-4464}.

Quotients are related to \emph{decomposition} of processes and
properties, an issue which has received considerable attention through
the years.
% There are many quotienting and decomposition techniques for various
% formalisms
In~\cite{DBLP:conf/lics/LarsenX90}, a solution to bisimulation $C(X)\sim
P$ for a given process $P$ and context $C$ is provided (as a~DMTS). This
solves the quotienting problem $P\by C$ for the special case where
both $P$ and $C$ are processes.
% the may transitions are identical with must transitions
This is extended in~\cite{DBLP:conf/icalp/LarsenX90} to the setting
where the context $C$ can have several holes and $C(X_1,\ldots,X_n)$
must satisfy a property $Q$ of \rhml.  However, $C$ remains to be a
process context, not a~specification context. Our \emph{specification}
context allows for arbitrary specifications, representing infinite sets
of processes and process equations.
% An~equation $C(X) \sim S$ where $S$ is a~specification and $C$ is
% a~specification context thus describes an~infinite set of process
% equations.
Another extension uses infinite
conjunctions~\cite{DBLP:journals/tcs/FokkinkGW06},
% , probabilistic processes~\cite{DBLP:conf/concur/GeblerF12} or
% processes with continuous time and
% space~\cite{DBLP:conf/icalp/CardelliLM11}.
but similarly to the other approaches, generates partial specifications
from an overall specification and a given set of processes.  This is
subsumed by a general quotient.
% However, these approaches provide decomposition rules for generating
% constraints $Q_1,\ldots,Q_n$ such that each $X_i$ satisfies $Q_i$ iff
% $X_1\parallel\cdots\parallel X_n$ satisfies $Q$. This does not induce
% a solution to the quotient, because instead of one of $X_i$ we need to
% put a given specification representing infinitely many processes.

Quotient operators, or \emph{guarantee} or \emph{multiplicative
  implication} as they are called there, are also well-known from
various logical formalisms.  Indeed, the algebraic properties of our
parallel composition $\|$ and quotient $\by$ resemble closely those of
multiplicative conjunction $\&$ and implication $\multimap$ in
\emph{linear logic}~\cite{DBLP:journals/tcs/Girard87}, and of spatial
conjunction and implication in \emph{spatial
  logic}~\cite{DBLP:journals/iandc/CairesC03} and \emph{separation
  logic}~\cite{DBLP:conf/lics/Reynolds02, DBLP:conf/csl/OHearnRY01}.
For these and other logics, proof systems have been developed which
allow one to reason about expressions containing these operators.

In spatial and separation logic, $\&$ and $\multimap$ (or the operators
corresponding to these linear-logic symbols) are first-class operators
on par with the other logical operators, and their semantics are defined
as certain sets of processes.  In contrast, for NAA and hence, via the
translations, also for \rhml, $\|$ and $\by$ are \emph{derived}
operators, and we provide constructions to reduce any expression which
contains them, to one which does not.  This is important from the
perspective of reuse of components and useful in industrial
applications.

% Noting that our operators obey most of the algebraic laws for linear
% logic~\cite{DBLP:journals/jsyml/Yetter90}, our paper should make the
% algebraic techniques from linear logic available for reasoning about
% NAA and \rhml specifications.

% Further, quotient operators have been used in
% spatial~\cite{DBLP:journals/iandc/CairesC03} or
% separation~\cite{DBLP:conf/lics/Reynolds02,DBLP:conf/csl/OHearnRY01}
% logics.  
% In contrast, we give a~constructive reduction of the quotient operator
% to expressions not containing it.
% %However, 
% To the best of our knowledge, 
% there are no such \emph{reductions} of quotient for the synchronisation type 
% of structural composition in the context of specifications. 
% The only such quotients that are moreover not limited to
% processes but also cover specifications (and also work in a graphically
% representable framework) are of~\cite{DBLP:journals/entcs/Raclet08}
% limited to the much simpler case of deterministic systems, which we
% extend in this paper.

% \textbf{Outline of the paper.} Section~\ref{sec:equiv} introduces all
% the formalisms and reveals their equivalence. Section~\ref{sec:lattice}
% deals with conjunction, disjunction and structural composition. In
% Section~\ref{sec:quotient} the quotient is constructed and in
% Section~\ref{sec:complement} complement and difference are discussed
% including the account on mistakes in the
% literature. Section~\ref{sec:conclusion} concludes.

\section{Specification Formalisms}
% and Their Equivalence}
\label{sec:equiv}

In this section, we define the specification formalisms \rhml, DMTS and NAA and show that they are equivalent.

% In order to show that disjunctive modal transition systems (DMTS) are as
% expressive as Hennessy-Milner logic with maximal fixed points (\rhml),
% we provide a chain of equivalences
% $$\text{\rhml} \Longleftrightarrow \text{BFS} \Longleftrightarrow \text{DMTS}$$
% Further, we show that BFS are equivalent to Boolean modal transition
% systems~\cite{?}. The equivalence of \rhml and $\nu$-calculus is well
% known~\cite{?}. As a result, we shall deduce that the logical formalisms
% can be ``drawn'' and possess a structural composition, and DMTS and BFS
% are members of the MTS family expressive enough to cover conjunction and
% disjunction.

For the rest of the paper, we fix a~finite alphabet $\Sigma$.  In each
of the formalisms, the semantics of a~specification is a~set of
implementations, in our case always a~set of \emph{labelled transition
  systems} (LTS) over $\Sigma$, \ie~structures $( S, s^0, \omust)$
consisting of a~set $S$ of \emph{states}, an~initial
state $s^0 \in S$, and a~\emph{transition relation}
$\omust\subseteq S\times \Sigma\times S$. 
We assume that the transition relation of LTS is always \emph{image-finite}, 
i.e.~that for every $a \in \Sigma$ and $s \in S$ the set 
$\{ s' \in S \mid s \must{a} s' \}$ is finite.

\subsection{Hennessy-Milner Logic with Greatest Fixed Points}

We recap the syntax and semantics of HML with variables developed
in~\cite{DBLP:journals/tcs/Larsen90}.  A \emph{HML formula} $\phi$
over a set $X$ of variables is given by the abstract syntax
$\phi\Coloneqq \ttt\mid \ff\mid x\mid \phi\land \phi\mid \phi\lor
\phi\mid \langle a\rangle \phi\mid[ a] \phi$, where $x$ ranges over
$X$ and $a$ over $\Sigma$.  The set of such formulae is denoted $\HML(
X)$. Notice that instead of including fixed point operators in the
logic, we choose to use declarations with a greatest fixed point
semantics, as explained below.

A \emph{declaration} is a mapping $\Delta: X\to \HML( X)$.  We shall
give a~greatest fixed point semantics to declarations.  Let $( S, s^0,
\omust)$ be an LTS, then an \emph{assignment} is a mapping
$\sigma: X\to 2^S$.  The set of assignments forms a complete lattice
with $\sigma_1\sqsubseteq \sigma_2$ iff $\sigma_1( x)\subseteq \sigma_2(
x)$ for all $x\in X$ and $\big(\bigsqcup_{ i\in I} \sigma_i\big)( x)=
\bigcup_{ i\in I} \sigma_i( x)$.

The semantics of a formula is a subset of $S$, given relative to an
assignment~$\sigma$, defined as follows: $\sem \ttt \sigma= S$, $\sem
\ff \sigma= \emptyset$, $\sem x \sigma= \sigma( x)$, $\sem{ \phi\land
  \psi} \sigma= \sem \phi\sigma\cap \sem \psi \sigma$, $\sem{ \phi\lor \psi}
\sigma= \sem \phi\sigma\cup \sem \psi \sigma$, $\sem{\langle a\rangle \phi}
\sigma=\{ s\in S\mid \exists s\must a s': s'\in \sem \phi \sigma\}$, and
$\sem{[ a] \phi} \sigma=\{ s\in S\mid \forall s\must a s': s'\in \sem
\phi \sigma\}$.  The semantics of a declaration $\Delta$ is then the
assignment defined by $\sem \Delta= \bigsqcup\{ \sigma: X\to 2^S\mid
\forall x\in X: \sigma( x)\subseteq \sem{ \Delta( x)} \sigma\}$: the
greatest (pre)fixed point of $\Delta$.

An \emph{initialised} HML declaration, or \emph{\rhml formula}, is a
structure $( X, X^0, \Delta)$, with $X^0\subseteq X$ finite sets of
variables and $\Delta: X\to \HML( X)$ a declaration.  We say that an~LTS
$( S, s^0, \omust)$ \emph{implements} (or models) the formula, and write
$S\models \Delta$, if it holds that %% for all $s^0\in S^0$, 
there is $x^0\in X^0$ such that $s^0\in \sem \Delta( x^0)$.  We write $\impl
\Delta$ for the set of implementations (models) of a \rhml formula
$\Delta$.

\subsection{Disjunctive Modal Transition Systems}

A DMTS is essentially a labelled transition system (LTS) with two types
of transitions, \textit{may} transitions which indicate that
implementations are permitted to implement the specified behaviour, and
\textit{must} transitions which proclaim that any implementation is
required to implement the specified behaviour.  Additionally, 
%all \textit{must} transitions are also \textit{may} transitions, and
\textit{must} transitions may be \emph{disjunctive}, in the sense that
they can require that \emph{at least one} out of a number of specified
behaviours must be implemented. We now recall the syntax and semantics of DMTS as introduced
in~\cite{DBLP:conf/lics/LarsenX90}.  We modify the syntax slightly to
permit multiple initial states and, in the spirit of later
work~\cite{DBLP:journals/entcs/FecherS05,DBLP:conf/atva/BenesCK11},
ensure that all required behaviour is also allowed:

A \emph{disjunctive modal transition system} (DMTS) over the alphabet 
$\Sigma$ is a~structure 
$( S, S^0, \omay, \omust)$ consisting of a~set of \emph{states} $S$,
a~finite subset $S^0\subseteq S$ of \emph{initial states}, 
a \emph{may}-transition relation $\omay\subseteq
S\times \Sigma\times S$, and a \emph{disjunctive must}-transition
relation $\omust\subseteq S\times 2^{ \Sigma\times S}$.  It is assumed
that for all $( s, N)\in \omust$ and all $( a, t)\in N$, $( s, a, t)\in
\omay$.  We usually write $s\may a t$ instead of $( s, a, t)\in \omay$
and $s\must{} N$ instead of $( s, N)\in \omust$.
We also assume that the may transition relation is image-finite.
Note that the two assumptions imply that 
$\omust\subseteq S\times \Pfin{ \Sigma\times S}$
where \Pfin{X}\ denotes the set of all finite subsets of $X$.

A DMTS $( S, S^0, \omay, \omust)$ is an \emph{implementation} if
$S^0 = \{ s^0\}$ is a~singleton and 
$\omust=\{( s,\{( a, t)\}\mid s\may a t\}$, hence if $N$ is a singleton
for each $s\must{} N$ and there are no superfluous may-transitions.
Thus DMTS implementations are precisely LTS.

We proceed to define the semantics of DMTS.  First, a relation
$R\subseteq S_1\times S_2$ is a \emph{modal refinement} between DMTS $(
S_1, S^0_1, \omay_1, \omust_1)$ and $( S_2, S^0_2, \omay_2, \omust_2)$
if it holds for all $( s_1, s_2)\in R$ that
\begin{itemize}
\item for all $s_1\may a t_1$ there is $s_2\may a t_2$ for some $t_2\in
  S_2$ with $( t_1, t_2)\in R$, and
\item for all $s_2\must{} N_2$ there is $s_1\must{} N_1$ such that for
  each $( a, t_1)\in N_1$ there is $( a, t_2)\in N_2$ with $( t_1,
  t_2)\in R$.
\end{itemize}
Such a modal refinement is \emph{initialised} if it is the case that,
for each $s^0_1\in S^0_1$, there is $s^0_2\in S^0_2$ for which $( s^0_1,
s^0_2)\in R$.  In that case, we say that $S_1$ refines $S_2$ and write
$S_1\mr S_2$.  We write $S_1\mreq S_2$ if $S_1\mr S_2$ and $S_2\mr S_1$.

We say that an LTS $I$ \emph{implements} a DMTS $S$ if $I\mr S$ and write
$\impl S$ for the set of implementations of $S$.  Notice that the
notions of implementation and modal refinement agree, capturing the
essence of DMTS as a \emph{specification theory}: A DMTS may be
\emph{gradually} refined, until an LTS, in which all behaviour is fully
specified, is obtained.

For DMTS $S_1$, $S_2$ we say that $S_1$ \emph{thoroughly} refines $S_2$,
and write $S_1\tr S_2$, if $\impl{ S_1}\subseteq \impl{ S_2}$.  We write
$S_1\treq S_2$ if $S_1\tr S_2$ and $S_2\tr S_1$.  By transitivity,
$S_1\mr S_2$ implies $S_1\tr S_2$.

\begin{example}
Figs.~\ref{fig:invariance} and~\ref{fig:until} show examples of important
basic properties expressed both as \rhml formulae, NAA (see below) and DMTS.  For
DMTS, may transitions are drawn as dashed arrows and disjunctive must transitions as branching arrows.
States with a short incoming arrow are initial (the DMTS in
Fig.~\ref{fig:until} has \emph{two} initial states).

\begin{figure}
%\figspacee
  \centering
  \begin{minipage}{.6\linewidth}
    $X= \langle a\rangle \ttt\land[ a] X\land[ b] X$

    \bigskip
	\scriptsize

	$(\{s_0\},\{s_0\},\Tran)$

	$\Tran(s_0) = \big\{ \{(a,s_0)\}, \{(a,s_0),(b,s_0)\}\big\}$
%%  $\mathord{\shortrightarrow} s_0\to( a, s_0)\land \neg( a, s_1)\land\neg( b,
%%  s_1)$

%%  $\phantom{\shortrightarrow} s_1\to \ttt$
  \end{minipage}
  \qquad \quad 
  \begin{minipage}{.2\linewidth}
    \begin{tikzpicture}[font=\footnotesize,->,>=stealth',scale=1.5,
      state/.style={shape=circle,draw,font=\scriptsize,inner
        sep=.5mm,minimum size=.22cm,outer sep=0.8mm,initial text=,initial
        distance=1.4ex}]
      \node[state,initial] (s) {};
      \path (s) edge[loop above] node[right] {$a$} (s);
      \path (s) edge[loop right,densely dashed] node[right]
      {$b$} (s);
    \end{tikzpicture}
  \end{minipage}
  \caption{%
    \label{fig:invariance}
    \rhml formula, NAA and DMTS for the invariance property
    ``\textit{there is always an `a' transition available}'', with
    $\Sigma=\{ a, b\}$}
\figspacee
\end{figure}
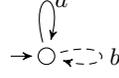

\begin{figure}
%\figspaceb
  \centering
  \begin{minipage}{.6\linewidth}
    $X= \langle b\rangle \ttt\lor\big( \langle a\rangle \ttt\land[ a]
    X\land[ b] X\land[ c] X\big)$

    \bigskip
	\scriptsize
	$(\{s_0,s_1\},\{s_0\},\Tran)$

	$\Tran(s_0) = \big\{
		\{(b,s_1)\}, 
		\{(b,s_1),(a,s_1)\},
		\{(b,s_1),(c,s_1)\},$

$	\phantom{\Tran(s_0) = \big\{ }		\{(b,s_1),(a,s_1),(c,s_1))\},
		\{(a,s_0)\},
		\{(a,s_0),(c,s_0)\}
	\big\}$

	$\Tran(s_1) = 2^{\{s_1\} \times \{a,b,c\}}$
	
%%    $\mathord{\shortrightarrow} s_0\to( b, s_1)\lor\big(( a, s_0)$ \\
%%    \mbox{}\hfill $\land \neg( a, s_1)\land \neg( c, s_1)\big)$
%%
%%    $\phantom{\shortrightarrow} s_1\to \ttt$
  \end{minipage}
  \qquad \quad
  \begin{minipage}{.3\linewidth}
    \begin{tikzpicture}[font=\footnotesize,->,>=stealth',scale=1,
      state/.style={shape=circle,draw,font=\scriptsize,inner
        sep=.5mm,outer sep=0.8mm,minimum size=.22cm,initial
        text=,initial distance=2ex}]
      \node[state,initial] (s01) at (0,0) {};
      \node[state,initial] (s02) at (0,-2) {};
      \node[state] (s1) at (1.5,-1) {};
%%      \coordinate (t2) at (.6,-1);
%%      \path (s02) edge[-,thick] (t2);
%%      \path (t2) edge node[left,pos=.2] {$a$} (s01);
%%      \path (t2) edge[out=-75,in=30] node[right,pos=.3] {$\!a$} (s02);
      \path (s02) edge[loop,in=25,out=75,looseness=15] 
	node[inner sep=0,outer sep=0,minimum size=0,pos=0.2,name=X] {} 
	node[right] {$a$} (s02);
	\path (X) edge[bend right] node[right,pos=0.4] {$\!a$} (s01);

      \path (s01) edge node[above] {$b$} (s1);
      \path (s01) edge[out=15,in=105,densely dashed] node[right] {$a, c$}
      (s1);
      \path (s02) edge[densely dashed] node[left] {$b, c$}
      (s01);
      \path (s02) edge[densely dashed,loop below] node[right]
      {$b, c$} (s02);
      \path (s1) edge[densely dashed,loop below] node[below]
      {$a, b, c$} (s1);
    \end{tikzpicture}
  \end{minipage}
  \caption{%
    \label{fig:until}
    \rhml formula, NAA and DMTS for the (``weak until'') property
    ``\textit{there is always an `a' transition available, until a
      `b' transition becomes enabled}'', with $\Sigma=\{ a, b, c\}$}
%\figspacee
\end{figure}
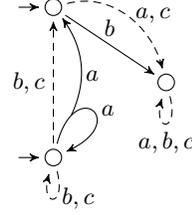
 
\end{example}

\subsubsection{Modal Transition Systems}

\begin{wrapfigure}{r}{0.3\textwidth}
  \begin{tikzpicture}[font=\footnotesize,->,>=stealth',scale=.7,
      state/.style={shape=circle,draw,font=\scriptsize,inner
        sep=.5mm,minimum size=.22cm,outer sep=0.8mm,initial text=,initial
        distance=3ex}]
      \node[state,initial] (s) at (0,0) {};
      \node[state] (t) at (2,.5) {};
      \node[state] (u) at (2,-.5) {};
      \coordinate (t1) at (1,0);
      \path[-] (s) edge (t1);
      \path (t1) edge node[above] {$a$} (t);
      \path (t1) edge node[below] {$b$} (u);
  \end{tikzpicture}
%   \caption{%
%     \label{fig:mtsex}
%     Simple DMTS for which no equivalent MTS exists}
\figspaceb
\end{wrapfigure}
An interesting subclass of DMTS are \emph{modal transition
  systems} (MTS)~\cite{DBLP:conf/avmfss/Larsen89}. A DMTS $( S, S^0,
\omay, \omust)$ is said to be a MTS if (1) $S^0=\{ s^0\}$ is a
singleton, (2) for every $s\must{} N$, the set $N$ is a
singleton. Hence, for each transition, we specify whether it must, may,
or must not be present; no disjunctions can be expressed.  
It is easy to
see that MTS are less expressive than DMTS, \ie~there are DMTS $S$ for
which no MTS $S'$ exists so that $\impl S= \impl{ S'}$.  One example is
provided on the right. %in Fig.~\ref{fig:mtsex}; 
Here any implementation must have an
$a$ or a $b$ transition from the initial state, but then any MTS which
permits all such implementations will also allow implementations without
any transition from the initial state.

\subsection{NAA}

We now define NAA, the nondeterministic extension to the formalism of
acceptance automata~\cite{DBLP:journals/entcs/Raclet08}.  We shall use
this formalism to bridge the gap between \rhml and
DMTS. A~\emph{nondeterministic acceptance automaton} over the alphabet
$\Sigma$ is a~structure $(S,S^0,\Tran)$ where $S$ and $S^0$ are the
states and initial states as previously, and $\Tran : S \to 2^{ \Pfin{
    \Sigma\times S}}$ assigns admissible transition sets.

% We interpret such a set $\Tran( s)=\{ M_1,\dots, M_n\}$ as a
% \emph{disjunction} $M_1\lor\dots\lor M_n$, and every single $M_i=\{(
% a_{ i1}, t_{ i1}),\dots,( a_{ in_i}, t_{ in_i})\}$ as a
% \emph{conjunction} $( a_{ i1}, t_{ i1})\land\dots\land( a_{ in_i}, t_{
% in_i})$.

% normal form which is so convenient that we equip it
% with its own name: A \emph{conjunctive modal transition system} (CMTS)
% is a structure $( S, S^0, \Tran)$, with $S^0\subseteq S$ finite sets as
% above and $\Tran: S\to 2^{ 2^{ \Sigma\times S}}$ a mapping which assigns
% sets of sets of admissible transitions to every state.
%
% This defines a translation from CMTS to BFS (in disjunctive normal
% form).  The reverse translation, from BFS to CMTS, is given by, for
% each $s\in S$, letting $\Tran( s)=\{ M\subseteq \Sigma\times S\mid
% M\models \Phi( s)\}$.

A~NAA $( S, S^0, \Tran)$ is an \emph{implementation} if 
$S^0 = \{s^0\}$ is a~singleton and $\Tran( s)=\{
M\}$ is a singleton for every $s\in S$; clearly, NAA implementations are
precisely LTS.  We also define the \emph{inconsistent NAA} to be $\bot=(
\emptyset, \emptyset, \emptyset)$ and the \emph{universal NAA} by
$\top=(\{ s\},\{ s\},  2^{ 2^{ \Sigma\times\{ s\}}})$.
% with $\Phi( s)= \ttt$.%  note the overloading
% of the $\bot$ and $\top$ symbols which should cause no
% confusion.

A relation $R\subseteq S_1\times S_2$ is a \emph{modal refinement}
between NAA $( S_1, S^0_1, \Tran_1)$, $( S_2, S^0_2, \Tran_2)$ if it holds
for all $( s_1, s_2)\in R$ and all $M_1\in \Tran_1( s_1)$ that there
exists $M_2\in \Tran_2( s_2)$ such that
\begin{itemize}
\item $\forall ( a, t_1)\in M_1: \exists ( a, t_2)\in M_2:( t_1, t_2)\in
  R$,
\item $\forall ( a, t_2)\in M_2: \exists ( a, t_1)\in M_1:( t_1, t_2)\in
  R$.
\end{itemize}
We define and use the notions of initialised modal refinement, $\mr$,
$\mreq$, implementation, $\tr$, and $\treq$ the same way as for DMTS.

% Similarly to DMTS, a modal refinement is \emph{initialised} if it is the
% case that, for each $s^0_1\in S^0_1$, there is $s^0_2\in S^0_2$ for
% which $( s^0_1, s^0_2)\in R$.  In that case, we say that $S_1$ refines
% $S_2$ and write $S_1\mr S_2$.  We write $S_1\mreq S_2$ if $S_1\mr S_2$
% and $S_2\mr S_1$.

\begin{proposition}
  \label{pr:bfs.preorder}
  The class of NAA is preordered by modal refinement $\mr$, with bottom
  element $\bot$ and top element $\top$.
\end{proposition}

% We say that an LTS $I$ \emph{implements} a NAA $S$ if $I\mr S$ and write
% $\impl S$ for the set of implementations of $S$.  
%% ; Theorem~\ref{th:tr-complete} in
%% Section~\ref{sec:complement} will show that for \emph{deterministic} BFS
%% (which we introduce later), the converse also holds.

Note that as implementations of all our three formalisms \rhml, DMTS and
NAA are LTS, it makes sense to use thorough refinement
$\tr$ and equivalence $\treq$ \emph{across} formalisms, so that we
\eg~can write $S\tr \Delta$ for a NAA $S$ and a \rhml formula $\Delta$.

\subsection{Equivalences}

We proceed to show that \rhml, DMTS and NAA are equally expressive:
%  in
% the sense of Theorem~\ref{th:equivalence} below.  To this end, we
% provide translations between DMTS and NAA, from NAA to \rhml, and from
% \rhml to DMTS.  Section~\ref{se:complexity} will be concerned with the
% complexity of the translations.

\begin{theorem}
  \label{th:equivalence}
  For any set $\mathcal S$ of LTS, the following are
  equivalent:
  \begin{enumerate}
  \item There exists a \rhml formula $\Delta$ with $\impl \Delta=
    \mathcal S$.
  \item There exists a finite NAA $S$ with $\impl S= \mathcal S$.
  \item There exists a finite DMTS $S$ with $\impl S= \mathcal S$.
  \end{enumerate}
  Furthermore, the latter two statements are equivalent even if we drop 
the finiteness constraints.
\end{theorem}
Note that we could drop the finiteness assumption about the set of variables
of \rhml formulae, while retaining the fact that $\Delta(x)$ is a~finite
HML formula. The result of Theorem~\ref{th:equivalence} could then be extended
with the statement that these possibly infinite \rhml formulae are
equivalent to general DMTS/NAA.

% \medskip \noindent {\bf Translation between DMTS and NAA.}
For a DMTS $S=( S, S^0, \omay, \omust)$,
%% define $\Phi: S\to \Bool( \Sigma\times S)$ by $\Phi( s)=
%%\bigland_{ s\must{} N} \biglor_{( a, t)\in N}( a, t)\land \bigland_{( s,
%%  a, t)\notin \omay} \neg( a, t)$
let $\Tran(s) =
\{ M \subseteq \Sigma \times S \mid \exists N: s\must{} N, N \subseteq M;
\forall (a,t) \in M: s \may{a} t \}$
% \begin{equation*}
%   \Phi( s)= \bigland_{ s\must{} N} \biglor_{( a, t)\in N}( a, t)\;\land\!\!
%   \bigland_{( s, a, t)\notin \omay}\!\! \neg( a, t)
% \end{equation*}
and define the NAA $\db( S)=( S, S^0, \Tran)$.

Conversely, for an NAA $( S, S^0, \Tran)$, define the DMTS $\bd( S)=( T,
T^0, \omay,\linebreak[4] \omust)$ as follows:
\begin{itemize}
\item $T=\{ M\in \Tran( s)\mid s\in S\}$,
 $T^0=\{ M\in \Tran( s^0)\mid s^0\in S^0\}$,
\item $\omust=\{( M,\{( a, M')\mid M'\in \Tran( s')\}\mid( a, s')\in
  M\}$,
\item $\omay=\{( t, a, t')\mid t\in T, \exists( t, N)\in \omust: ( a,
  t')\in N\}$.
\end{itemize}

Note that both $\bd$ and $\db$ preserve finiteness. Both translation 
are exponential in their respective arguments.

\begin{lemma}
  \label{th:bfsvsdmts}
  For every DMTS $S$, $S\treq \db( S)$. For every NAA $S$, $S\treq \bd( S)$.
\end{lemma}

% \subsection{Equivalence of BFS and \rhml}

% For a BFS $( S, S^0, \Phi)$, define $\bh( S)=( S, S^0, \Delta)$ with
% $\Delta( s)=\biglor_{ M\in \Tran( s)} \bigland_{( a, t)\in M} \langle
% a\rangle t$.

%% BFS have a convenient semantic representation using sets of
%% \emph{transition constraints} $\Tran: S\to 2^{ 2^{ \Sigma\times S}}$.
%% % (where $2^X$ denotes the power set of a given set $X$).
%% For $s\in S$, the constraint is given by $\Tran( s)=\{ M\subseteq
%% \Sigma\times S\mid M\models \Phi( s)\}$.  Conversely, we have $\Phi(
%% s)= \biglor_{ M\in \Tran( s)}\big( \bigland_{( a, t)\in M}( a, t)\land
%% \bigland_{( a, t)\notin M}\neg( a, t)\big)$, so that the
%% representations of BFS by obligation functions and by transition
%% constraints are equivalent.

% \noindent {\bf Translations between finite NAA and \rhml.} 
For a~set of pairs of actions and states $M$ we use 
$M_a$ to denote the set $\{ s \mid (a,s)\in M\}$.
Let $( S, S^0, \Tran)$ be a~finite NAA and let $s \in S$,
%% For a state $s$, 
we then define
\[\Delta_\Tran(s) = \biglor_{ M\in \Tran( s)}\big( \bigland_{( a, t)\in M}
\langle a\rangle t\land
 \bigland_{ a\in \Sigma}
[a]
\big(\biglor_{ u \in M_a} u  \big)
\big)\]

%% For the translation from BFS to \rhml, we first define a mapping $\bhf:
%% \Bool( \Sigma\times S)\to \HML( S)$ from Boolean to \rhml formulae, by
%% induction on the structure of the formula.  Using de~Morgan's laws, we
%% can assume that only atomic terms $( a, t)$ are negated.  We let
%% \begin{itemize}
%% \item $\bhf( \ttt)= \ttt$, $\bhf( \ff)= \ff$,
%%  $\bhf( \phi_1\land \phi_2)= \bhf( \phi_1)\land \bhf( \phi_2)$,
%%  $\bhf( \phi_1\lor \phi_2)= \bhf( \phi_1)\lor \bhf( \phi_2)$,
%% \item $\bhf(( a, t))= \langle a\rangle t$,
%% \item $\bhf( \neg( a, t))= [ a]\big( \biglor_{ u\in S\setminus\{
%%     t\}}u\big)$.
%% \end{itemize}
%% Now let $( S, S^0, \Phi)$ be a BFS in which only atomic terms are
%% negated in each formula $\Phi( s)$ and define 
We then define the \rhml formula $\bh( S)=( S, S^0,
\Delta_\Tran)$. Notice that variables in $\bh(S)$ are states of $S$.
%% with $\Delta( s)= \bhf( \Phi( s))$ for each $s\in S$.
%
% For the reverse translation, we use the normal form from
% Lemma~\ref{le:hmlnormal}.
% % Let $( X, X^0, \Delta)$ be an initialized declaration in the form
% % introduced in the lemma, then we define a DMTS $\hd( \Delta)=( X, X^0,
% % \omay, \omust)$ as follows:
% % \begin{itemize}
% % \item If $\Delta( x)= \top$, let $\omay( x)= 2^{ \Sigma\times\{ x\}}$
% %   and $\omust( x)= \emptyset$.
% % \item If $\Delta( x)= \biglor_{ i\in I}( \bigland_{ j\in J_i} \langle
% %   a_{ ij}\rangle x_{ ij}\land \bigland_{ a\in \Sigma}[ a] y_{ i, a})$,
% %   let $\Succ( x)=\{\{( a_{ ij}, x_{ ij})\mid j\in J_i\}\mid i\in
% %   I\}$ and $\omay( x)=$
% % \end{itemize}
% Let $( X, X^0, \Delta)$ be an initialized HML declaration in the
% form introduced in the lemma, then we define a BFS $\hb( \Delta)=( X,
% X^0, \Phi)$ as follows:
% \begin{itemize}
% \item If $\Delta( x)= \top$, let $\Phi( x)= \top$.
% \item If $\Delta( x)= \biglor_{ i\in I}\big( \bigland_{ j\in J_i} \langle
%   a_{ ij}\rangle x_{ ij}\land \bigland_{ a\in \Sigma}[ a] y_{ i, a}\big)$,
%   let $\Phi( x)= \biglor_{ i\in I}\big( \bigland_{ j\in J_i} \biglor\{( a_{
%     ij}, z_{ ij})\mid \impl{ \Delta( z_{ ij})}\subseteq \impl{ \Delta(
%     x_{ ij})}\}\land \bigland_{ a\in \Sigma} \bigland\{ \neg( a, z_{ a,
%     i})\mid \impl{ \Delta( z_{ a, i})}\not\subseteq \impl{ \Delta( y_{
%       a, i})}\}\big)$.
% \end{itemize}

\begin{lemma}
  \label{th:bfsvshml}
  For all NAA $S$, $S\treq \bh( S)$.
  % For all initialized HML declarations $\Delta$, $\Delta\treq \hb(
  % \Delta)$.
\end{lemma}

% \subsection{Direct Translation from \rhml to DMTS}

Our translation from \rhml to DMTS is based on the constructions
in~\cite{DBLP:journals/tcs/BoudolL92}.  First, we need a variant of a
disjunctive normal form for HML formulae:
% introduced in this paper:

\begin{lemma}%[\cf~\cite{DBLP:journals/tcs/BoudolL92}]
  \label{le:hmlnormalstrong}
  For any \rhml formula $( X_1, X^0_1, \Delta_1)$, there exists another
  formula $( X_2, X^0_2, \Delta_2)$ with $\impl{ \Delta_1}= \impl{
    \Delta_2}$ and such that any formula $\Delta_2( x)$, for $x\in X_2$,
  is $\ttt$ or of the form $\Delta_2( x)=\biglor_{ i\in I}\big(
  \bigland_{ j\in J_i} \langle a_{ ij}\rangle x_{ ij}\land \bigland_{
    a\in \Sigma}[ a] y_{ i, a}\big)$ for finite (possibly empty) index
  sets $I$ and $J_i$, $i\in I$, and all $x_{ ij}, y_{ i, a}\in X_2$.
  Additionally we can assume that for all $i\in I$, $j\in J_i$, $a\in
  \Sigma$, $a_{ ij}= a$ implies $\impl{ x_{ ij}}\subseteq \impl{ y_{ i,
      a}}$.
\end{lemma}
 
Let now $( X, X^0, \Delta)$ be a \rhml formula in the form introduced
above, then we define a DMTS $\hd( \Delta)=( S, S^0, \omay, \omust)$ as
follows:
\begin{itemize}
\item $S=\{( x, k)\mid x\in X, \Delta( x)= \biglor_{ i\in I} \phi_i,
  k\in I\ne \emptyset\}\cup\{ \bot, \top\}$,
\item $S^0=\{( x^0, k)\mid x^0\in
  X^0\}$.
\item For each $( x, k)\in S$ with $\Delta( x)=\biglor_{ i\in I}(
  \bigland_{ j\in J_i} \langle a_{ ij}\rangle x_{ ij}\land \bigland_{
    a\in \Sigma}[ a] y_{ i, a})$ and $I\ne \emptyset$,
  \begin{itemize}
  \item for each $j\in J_i$, let $\Must_j( x, k)=\{( a_{ ij},( x_{ ij},
    i'))\in \Sigma\times S\}$,
  \item for each $a\in \Sigma$, let $\May_a( x, k)=\{( x', i')\in S\mid
    \impl{ x'}\subseteq \impl{ y_{ i, a}}\}$.
  \end{itemize}
\item Let $\omay=\{( s, a, s')\mid s\in S, a\in \Sigma, s'\in \May_a(
  s)\}\cup\{( \top, a, \top)\mid a\in \Sigma\}$ and $\omust=\{( s,
  \Must_j( s))\mid s=( x, i)\in S, j\in J_i\}\cup\{( \bot,
  \emptyset)\}$.
\end{itemize}

\begin{lemma}
  \label{th:hmltodmts}
  For all \rhml formulae $\Delta$, $\Delta\treq \hd( \Delta)$.
\end{lemma}

Further, we remark that the overall translation from DMTS to \rhml is quadratic and in the other direction inevitably exponential.

\begin{example}\label{ex:rhml-to-dmts}
%%Re-consider the \rhml formula $X= \langle a\rangle
%%\ttt\land\left(([ a] \ff\land[ b]\ff) \lor \langle a\rangle X\lor
%%  \langle b\rangle X\right)$ from the example in Fig.~\ref{fig:safety}.\todo{J:is in appendix, put back and remove another one}
%
%%  Consider the \rhml formula $X= \langle a\rangle \ttt\land(([ a]
%%    \ff\land[ b]\ff) \lor \langle a\rangle X$$\lor$$\langle b\rangle
%%    X)$.  Distributing the term $\langle a \rangle \ttt$ over the
%%  rest of the formula yields the normal form of
%%  Lemma~\ref{le:hmlnormalstrong}.  
  Consider the \rhml formula 
	$X = (\langle a\rangle 
		(\langle b\rangle X \land [a] \ff ) 
	\land [b] \ff ) \lor [a]\ff$.
  Changing the formula into the normal form of Lemma~\ref{le:hmlnormalstrong}
  introduces a~new variable $Y$ as illustrated below; 
  $X$ remains the sole initial variable. The translation \hd\ then
  gives a~DMTS with two initial states %as follows
 (the inconsistent state~$\bot$ 
  and redundant may transitions such as $x_1 \may{a} x_2$, 
  $x_2 \may{b} x_1$, 
etc.~have been omitted):

%The conversion then gives a~DMTS with
%  three initial states and one extra universal state with characteristic
%  formulae as follows (the inconsistent initial state $x_1^0$ 
%  and redundant may transitions such as $x_2^0 \may{b} x_3^0$, $x_2^0 \may{b} x_2^0$, etc.~have been
%  removed): %%\todo{N: I do not understand this example.}

%% \begin{figure}
%% \figspacee
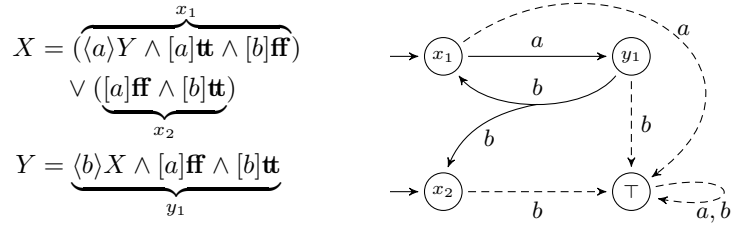
\centering
\begin{minipage}[b]{0.4\textwidth}
%\vspace*{-10em}
\begin{align*}
X &= (\overbrace{\langle a\rangle Y \land [a]\ttt \land [b]\ff}^{x_1}) \\
&\quad \lor(\underbrace{[a]\ff \land [b]\ttt}_{x_2}) \\
Y &= \underbrace{\langle b\rangle X \land [a]\ff \land [b]\ttt}_{y_1}
%  x_1^0 := \ff \qquad
%  x_2^0 := \langle a\rangle X \land [a]\ttt \land [b]\ttt \\
%  x_3^0 := \langle a \rangle \ttt \land \langle b \rangle X \land
%  [a]\ttt \land [b]\ttt \qquad
%  \top := \ttt
\end{align*}
\end{minipage}
%\hspace{.7cm}
    \begin{tikzpicture}[baseline=(bas),font=\footnotesize,->,>=stealth',yscale=1.8,xscale=2.5,
      state/.style={shape=circle,draw,font=\scriptsize,inner
        sep=.5mm,outer sep=0.8mm,minimum size=.5cm,initial
        text=,initial distance=1ex}]
	\path[use as bounding box] (-1,-1) rectangle (1,0);
	\coordinate (bas) at (0,-1.2);
      \node[state, initial left] (s0) at (-.5,0) {$x_1$};
      \node[state] (s1) at (.5,-1) {$\top$};
      \node[state, initial left] (s2) at (-.5,-1) {$x_2$};
      \node[state] (s3) at (.5,0) {$y_1$};
      \coordinate (t1) at (-.5,-.3);
      \coordinate (t2) at (0,-.6);
      \path (s0) edge[densely dashed,bend left,looseness=2,out=90,in=90] node[right] {$a$} (s1);
      \path (s0) edge node[above] {$a$} (s3);
%%      \path (s0) edge[-,thick] (t1);
%%      \path (t1) edge[bend right] node[right] {$a$} (s2);
%%      \path (t1) edge[bend right] node[below] {$a$} (s0);
      %\path (s0) edge[loop,in=235,out=205,looseness=15] 
%	node[inner sep=0,outer sep=0,minimum size=0,pos=0.2,name=X] {} 
%	node[below] {$a$} (s0);
%	\path (X) edge[bend right] node[left,pos=0.4] {$\!a$} (s2);

%%      \path (s2) edge[-,thick] (t2);
%%      \path (t2) edge[bend left] node[right] {$b$} (s0);
%%      \path (t2) edge[bend right] node[above,pos=.7] {$b$} (s2);
      \path (s3) edge [bend left=60] %[loop,in=55,out=25,looseness=15] 
	node[inner sep=0,outer sep=0,minimum size=0,pos=0.5,name=XX] {} 
	node[above] {$b$} (s0);
      \path (XX) edge[bend right] node[below,pos=0.4] {$b$} (s2);

      \path (s2) edge[densely dashed] node[below] {$b$} (s1);
%      \path (s2) edge node[above] {$a$} (s1);
      \path (s3) edge[densely dashed] node[right] {$b$} (s1);
      \path (s1) edge[loop right,densely dashed] node[below]
      {$a, b$\,\,\,\,\,} (s1);
    \end{tikzpicture}
%%   \caption{%
%%     \label{fig:ex-rhml-to-dmts}
%% %    DMTS translation of %% the safety 
%% %    a \rhml formula 
%% 	%% in    Fig.~\ref{fig:safety}
%%      Illustration of Example~\ref{ex:rhml-to-dmts}.
%% }
%% \figspaceb
%% \end{figure}
\end{example}

\section{Specification Theory}
\label{sec:lattice}

In this section, we introduce operations of conjunction, disjunction,
structural composition and quotient for NAA, DMTS and \rhml.  Together,
these operations yield a \emph{complete specification theory} in the
sense of~\cite{DBLP:conf/fase/BauerDHLLNW12}, which allows for
compositional design and verification using both logical and structural
operations.
We remark that conjunction and disjunction are straightforward for
logical formalisms such as \rhml, whereas structural composition is more
readily defined on behavioural formalisms such as (D)MTS.  For the mixed
formalism of NAA, disjunction is trivial as we permit multiple initial
states, but conjunction requires some work.
Note that our construction of conjunction works for nondeterministic
systems in contrast to all the work in this area except for 
\cite{DBLP:conf/atva/BenesCK11, DBLP:conf/lics/LarsenX90}.
% Structural composition has not previously been available for
% \rhml.\todo{is this true?}

\subsection{Disjunction}

The disjunction of NAA $S_1 = (S_1, S^0_1, \Tran_1)$ and $S_2 = (S_2,
S^0_2, \Tran_2)$ is $S_1\lor S_2=( S_1 \cup S_2, S^0_1 \cup S^0_2,
\Tran_1 \cup \Tran_2)$. Similarly, the disjunction of two DMTS $S_1 =
(S_1, S^0_1, \omay_1,\omust_1)$ and $S_2 = (S_2, S^0_2,
\omay_2,\omust_2)$ is $S_1\lor S_2= (S_1 \cup S_2, S^0_1 \cup S^0_2,
\omay_1 \cup \omay_2, \omust_1\cup\omust_2)$. It follows that
disjunction respects the translation mappings $\db$ and $\bd$ from the
previous section.

% \begin{lemma}
%   \label{th:disj=lub}
%   For NAA or DMTS $S_1$, $S_2$, $S_3$, $S_1\vee S_2\mr S_3$ iff $S_1\mr
%   S_3$ and $S_2\mr S_3$.
% \end{lemma}
% 
% \begin{theorem}
%   \label{th:disj-impl}
%   For NAA or DMTS $S_1$, $S_2$, $\impl{ S_1\lor S_2}= \impl{ S_1}\cup
%   \impl{ S_2}$.
% \end{theorem}

\begin{theorem}
  \label{th:disj-impl}
  Let $S_1$, $S_2$, $S_3$ be NAA or DMTS. Then $\impl{ S_1\lor S_2}=
  \impl{ S_1}\cup \impl{ S_2}$. Further, $S_1\vee S_2\mr S_3$ iff
  $S_1\mr S_3$ and $S_2\mr S_3$.
\end{theorem}

\begin{wrapfigure}{r}{0.4\textwidth}
  \figspaceb
  \begin{tikzpicture}[x=3cm,y=2cm,font=\footnotesize,->,>=stealth',scale=.75,
    state/.style={shape=circle,draw,font=\scriptsize,inner sep=.5mm,outer
      sep=0.8mm, minimum size=0.22cm,initial text=,initial
      distance=3ex}]
    \node[state,initial] (X) at (0,0) {$s_1^0$};
    \node[state,initial right] (Y) at (1,0) {$s_2^0$};
%%    \coordinate (Z) at (0.75,-0.15);
%%    \path[-,thick] (X) edge (Z); 
%%    \path[->] (Z) edge[bend left] node[below]{$a$} (X); 
%%    \path[->] (Z) edge[bend right] node[below]{$a$} (Y); 
      \path (X) edge[loop,in=-35,out=-5,looseness=15] 
	node[inner sep=0,outer sep=0,minimum size=0,pos=0.2,name=Xa] {} 
	node[below,pos=0.7] {$a$} (X);
      \path (Xa) edge[bend right] node[below,pos=0.6] {$a$} (Y);

%%    \coordinate (ZZ) at (0.75,0.15);
%%    \path[-,thick] (X) edge (ZZ); 
%%    \path[->] (ZZ) edge[bend right] node[above]{$b$} (X); 
%%    \path[->] (ZZ) edge[bend left] node[above]{$b$} (Y); 
      \path (X) edge[loop,in=35,out=5,looseness=15] 
	node[inner sep=0,outer sep=0,minimum size=0,pos=0.2,name=Xb] {} 
	node[above,pos=0.7] {$b$} (X);
      \path (Xb) edge[bend left] node[above,pos=0.6] {$b$} (Y);

  \end{tikzpicture}
%   \caption{%
%     \label{fig:two-init}
%     DMTS for which two initial states are necessary}
\figspacee
\end{wrapfigure}
We point out one important distinction between NAA and DMTS: NAA with a
\emph{single} initial state are equally expressive as general NAA, while
for DMTS, this is not the case. The example on the right %in Fig.~\ref{fig:two-init}
shows a DMTS $( S, S^0, \omay,
\omust)$, with $S= S^0=\{ s_1^0, s_2^0\}$, $s_1^0\must{}\{( a, s_1^0),(
a, s_2^0)\}$ and $s_1^0\must{}\{( b, s_1^0),( b, s_2^0)\}$ (and the
corresponding may-transitions).  Two initial states are necessary for
capturing $\impl S$.

\begin{lemma}
  \label{le:bfs.oneinitial}
  For any NAA $S$ there is a NAA $T=( T, T^0, \Psi)$ with $T^0=\{ t^0\}$
  a~singleton and $S\mreq T$.
\end{lemma}

% Given that we permit multiple initial states, disjunction of both DMTS
% and BFS is easily defined.  For DMTS $(S_1, S^0_1, \omay_1, \omust_1)$,
% $(S_2, S^0_2, \omay_2, \omust_2)$ with $S_1\cap S_2= \emptyset$, we
% define $S_1\lor S_2=( S_1\cup S_2, S^0_1\cup S^0_2, \omay_1\cup \omay_2,
% \omust_1\cup \omust_2)$.  Similarly, for BFS $( S_1, S^0_1, \Tran_1)$,
% $( S_2, S^0_2, \Tran_2)$ with $S_1\cap S_2= \emptyset$, we let $S_1\lor
% S_2=( S_1\cup S_2, S^0_1\cup S^0_2, \Tran_1\cup \Tran_2)$.

% \begin{lemma}
%   \label{le:disj-dmts}
%   For DMTS $S_1$, $S_2$, $\db( S_1\lor S_2)= \db( S_1)\lor \db(
%   S_2)$.
% \end{lemma}

% As a consequence of Lemma~\ref{le:disj-dmts}, the same property holds
% for DMTS, hence disjunction is the least upper bound (with respect to
% $\mr$) for both DMTS and BFS.

\subsection{Conjunction}
\label{sec:bmts-conjunction}

Conjunction for DMTS is an extension of the construction 
from~\cite{DBLP:conf/atva/BenesCK11}
for multiple initial states.  Given two DMTS $( S_1, S^0_1, \omay_1,
\omust_1)$, $( S_2, S^0_2, \omay_2, \omust_2)$, we define $S_1\land
S_2=( S, S^0, \omay, \omust)$ with $S= S_1\times S_2$, $S^0= S^0_1\times
S^0_2$, and %$\omay$ and $\omust$ given by
\begin{itemize}
\item $( s_1, s_2)\may a( t_1, t_2)$ iff $s_1\may a_1 t_1$ and $s_2\may
  a_2 t_2$,
\item for all $s_1\must{} N_1$, $( s_1, s_2)\must{} \{( a,( t_1,
  t_2))\mid( a, t_1)\in N_1,( s_1, s_2)\may a( t_1, t_2)\}$,
\item for all $s_2\must{} N_2$, $( s_1, s_2)\must{} \{( a,( t_1,
  t_2))\mid( a, t_2)\in N_2,( s_1, s_2)\may a( t_1, t_2)\}$.
\end{itemize}

To define conjunction for NAA, we need %% some 
auxiliary 
%% mappings on
%% Boolean formula.  For sets $S_1$, $S_2$ and $i\in\{ 1, 2\}$, we define
%% $\rho_i: \Bool( \Sigma\times S_i)\to \Bool( \Sigma\times S_1\times S_2)$
%% inductively, by
%% \begin{itemize}
%% \item $\rho_i( \ttt)= \ttt$, $\rho_i( \ff)= \ff$, $\rho_i(
%% \phi_i\land \phi_2)= \rho_i( \phi_i)\land \rho_i( \phi_2)$, $\rho_i(
%% \phi_i\lor \phi_2)= \rho_i( \phi_i)\lor \rho_i( \phi_2)$, $\rho_i( \neg
%% \phi)= \neg \rho_i( \phi)$, 
%% \item $\rho_1(( a, s_1))= \biglor_{ s_2\in S_2}( a, s_1, s_2)$,
%% \item $\rho_2(( a, s_2))= \biglor_{ s_1\in S_1}( a, s_1, s_2)$.
%% \end{itemize}
%% Then for NAA $( S_1, S_1^0, \Phi_1)$, $( S_2, S_2^0, \Phi_2)$, $S_1\land
%% S_2=( S_1\times S_2, S_1^0\times S_2^0, \Phi)$ with $\Phi(( s_1, s_2))=
%% \rho_1( \Phi_1( s_1))\land \rho_2( \Phi_2( s_2))$.
%% 
%% There is a similar construction which uses transition constraints
%%instead of obligation functions and 
projection functions $\pi_i:
\Sigma\times S_1\times S_2\to \Sigma\times S_i$.
%% instead of the $\rho_i$ above.  
These are defined by 
% $\pi_1( M)=\{( a, s_1)\mid \exists s_2\in
% S_2:( a, s_1, s_2)\in M\}$, $\pi_2( M)=\{( a, s_2)\mid \exists s_1\in
% S_1:( a, s_1, s_2)\in M\}$.  
\begin{align*}
\pi_1( M)=&\{( a, s_1)\mid \exists s_2\in
S_2:( a, s_1, s_2)\in M\}\\
\pi_2( M)=&\{( a, s_2)\mid \exists s_1\in
S_1:( a, s_1, s_2)\in M\}
 \end{align*}
%%Then for 
Given NAA $( S_1, S^0_1, \Tran_1)$, $(
S_2, S^0_2, \Tran_2)$, define $S_1\land S_2=( S, S^0, \Tran)$, with $S=
S_1\times S_2$, $S^0= S^0_1\times S^0_2$
%% as before, 
and $\Tran(( s_1,
s_2))=\{ M\subseteq \Sigma\times S_1\times S_2\mid \pi_1( M)\in \Tran_1(
s_1), \pi_2( M)\in \Tran_2( s_2)\}$.  
%% To see that the two conjunction
%% constructions agree, one only needs to notice that for $i\in\{ 1, 2\}$,
%% $M\models \rho_i( \phi)$ iff $\pi_i( M)\models \phi$.

% For $M_1\subseteq \Sigma\times S_1$ and $M_2\subseteq \Sigma\times S_2$,
% define $M_1\land M_2\subseteq 2^{ \Sigma\times S_1\times S_2}$ by
% $M_1\land M_2=\{ M\mid \pi_1( M)= M_1, \pi_2( M)= M_2\}$, where $\pi_1$,
% $\pi_2$ are the projection operators defined by $\pi_1( M)=\{( a,
% s_1)\mid \exists s_2\in S_2:( a, s_1, s_2)\in M\}$, $\pi_2( M)=\{( a,
% s_2)\mid \exists s_1\in S_1:( a, s_1, s_2)\in M\}$.  Now for BFS $( S_1,
% S^0_1, \Tran_1)$, $( S_2, S^0_2, \Tran_2)$, define $S_1\land S_2=( S,
% S^0, \Tran)$, with $S= S_1\times S_2$, $S^0= S^0_1\times S^0_2$, and
% $\Tran: S\to 2^{ 2^{ \Sigma\times S}}$ given by $\Tran( s_1,
% s_2)=\bigcup_{ M_1\in \Tran_1( s_1), M_2\in \Tran_2( s_2)} M_1\land
% M_2$.

\begin{lemma}
  \label{le:conj-dmts}
  For DMTS $S_1$, $S_2$, $\db( S_1\land S_2)= \db( S_1)\land \db( S_2)$.
  % For BFS $S_1$, $S_2$, $\bd( S_1\land S_2)= \bd( S_1)\land \bd(
  % S-2)$.
\end{lemma}

For the translation from NAA to DMTS, $\bd( S_1\land S_2)= \bd(
S_1)\land \bd( S_2)$ does not necessarily hold, as the translation
changes the state space.  However, Theorem~\ref{co:conj=glb} below will
ensure that $\bd(S_1 \land S_2)\treq \bd(S_1)\land \bd(S_2)$.

% \begin{lemma}
%   \label{th:conj=glb}
%   For NAA or DMTS $S_1$, $S_2$, $S_3$, $S_1 \mr S_2 \wedge S_3$ iff $S_1
%   \mr S_2$ and $S_1 \mr S_3$.
% \end{lemma}
% 
% \begin{theorem}\label{co:conj=glb}
%   For NAA or DMTS $S_1$, $S_2$, $\impl{ S_1\land S_2}= \impl{ S_1}\cap
%   \impl{ S_2}$.
% \end{theorem}

\begin{theorem}\label{co:conj=glb}
  Let $S_1$, $S_2$, $S_3$ be NAA or DMTS. Then $\impl{ S_1\land S_2}=
  \impl{ S_1}\cap \impl{ S_2}$. Further, $S_1 \mr S_2 \wedge S_3$ iff
  $S_1 \mr S_2$ and $S_1 \mr S_3$.
\end{theorem}

\begin{theorem}
  \label{th:lattice}
  With operations $\land$ and $\lor$, the sets of DMTS and NAA form
  bounded distributive lattices up to $\mreq$.
\end{theorem}

\subsection{Structural Composition}

We define structural composition for NAA. 
%% on the level of transition constraints.
For NAA $S_1=( S_1, S^0_1, \Tran_1)$, $S_2=( S_2, S^0_2,
\Tran_2)$, we define $S_1\| S_2=( S, S^0, \Tran)$ with $S= S_1\times
S_2$, $S^0= S^0_1\times S^0_2$, and for all $( s_1, s_2)\in S$, $\Tran((
s_1, s_2))=\{ M_1\| M_2\mid M_1\in \Tran_1( s_1), M_2\in \Tran_2(
s_2)\}$, where $M_1\| M_2=\{( a,( t_1, t_2))\mid( a, t_1)\in M_1,( a,
t_2)\in M_2\}$.

% Unraveling this definition, we see that for the Boolean formulae
% $\Phi( s)$, this means that $\|$ distributes over both $\lor$ and
% $\land$.

%Below we need the \emph{universal implementation}, which is the LTS
%$\mathsf{U}=(\{ s\}, s, \omust)$ with $s\must a s$ for all $a\in
%\Sigma$, or as NAA constraint, $\Tran( s)=\{ \Sigma\times\{ s\}\}$.

\begin{lemma}
  \label{le:bfs||prop}
  Up to $\mreq$, the operator $\|$ on NAA is associative and
  commutative, distributes over $\lor$, and has unit $\mathsf{U}$,
  where $\mathsf{U}$ is the LTS $(\{ s\}, s, \omust)$ with $s\must a s$ 
  for all $a\in \Sigma$.
\end{lemma}

% Hence BFS form, up to $\mreq$, a lattice-ordered monoid~\cite{?}.  In
% the next section we shall provide a residual to $\|$.

\begin{theorem}
  \label{th:bfs||indimp}
  For all NAA $S_1$, $S_2$, $S_3$, $S_4$, $S_1\mr S_3$ and $S_2\mr S_4$
  imply $S_1\| S_2\mr S_3\| S_4$.
\end{theorem}

We remark that structural composition on
MTS~\cite{DBLP:conf/avmfss/Larsen89} coincides with our NAA composition,
so that for MTS $S_1$, $S_2$, $\db( S_1)\| \db( S_2)= \db( S_1\| S_2)$.
On the other hand, structural composition for DMTS (with single initial
states) as defined in~\cite{DBLP:conf/atva/BenesCK11} is \emph{weaker}
than NAA composition, \ie~for DMTS $S_1$, $S_2$, and denoting by $\|'$
the composition from~\cite{DBLP:conf/atva/BenesCK11}, only $\db( S_1)\|
\db( S_2)\tr \db( S_1\|' S_2)$ holds. Consider for
example the DMTS $S$ and $S'$ in the figure below.
% Fig.~\ref{fig:dmts-compo}.
When considering their NAA composition, the initial state is the pair
$(s_0,t_0)$ with $\Tran((s_0,t_0)) = \{\emptyset,
\{(a,(s_2,t_1)),(a,(s_2,t_2))\}$. Since this constraint cannot be
represented as a disjunctive must,
% set does not satisfy the property of Lemma~\ref{le:dmtstobfsspecial},
there is no DMTS with a single initial state which can represent the NAA
composition precisely.

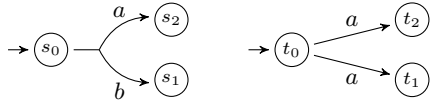
\begin{figure}
\figspacee
  \centering
  \begin{tikzpicture}[x=2cm,y=1cm,font=\footnotesize,->,>=stealth',scale=.8,
    state/.style={shape=circle,draw,font=\scriptsize,inner sep=.5mm,outer
      sep=0.8mm, minimum size=0.22cm,initial text=,initial
      distance=2.5ex}]
    \begin{scope}
    \node[state,initial] (X) at (0,0) {$s_0$};
    \node[state] (Y) at (1,-.5) {$s_1$};
    \node[state] (Z) at (1,.5) {$s_2$};
    \coordinate (XX) at (.4,0);
    \path[-] (X) edge (XX); 
    \path[->] (XX) edge[bend left] node[above]{$a$} (Z); 
    \path[->] (XX) edge[bend right] node[below]{$b$} (Y); 
    \end{scope}
    \begin{scope}[xshift=4cm]
    \node[state,initial] (X) at (0,0) {$t_0$};
    \node[state] (Y) at (1,-.5) {$t_1$};
    \node[state] (Z) at (1,.5) {$t_2$};
    \path[->] (X) edge node[above]{$a$} (Z); 
    \path[->] (X) edge node[below]{$a$} (Y); 
    \end{scope}
  \end{tikzpicture}
  % \caption{%
  %   \label{fig:dmts-compo}
  %   DMTS $S$, $S'$ whose DMTS composition is not precise}
\figspacee
\end{figure}

Hence the DMTS composition of~\cite{DBLP:conf/atva/BenesCK11} is a DMTS
over-approximation of the NAA composition, and translating from DMTS to
NAA before composing (and back again) will generally give a
tighter specification.  However, as noted already
in~\cite{DBLP:conf/ershov/HuttelL89}, MTS composition itself is an
over-approximation, in the sense that there will generally be
implementations $I\in \impl{ S_1\| S_2}$ which cannot be written $I=
I_1\| I_2$ for $I_1\in \impl{ S_1}$ and $I_2\in \impl{ S_2}$; the same
is the case for NAA and DMTS.

\subsection{Quotient}
\label{sec:quotient}

We now present one of the central contributions of this paper, the
construction of quotient. The quotient $S\by T$ is to be the most
general specification that, when composed with $T$, refines $S$. In
other words, it must satisfy the property that for all specifications
$X$, $X \mr S\by T$ iff $X \parallel T \mr S$.  Quotient has been
defined for deterministic MTS and for deterministic acceptance automata
in~\cite{DBLP:journals/entcs/Raclet08}; here we extend it to the
nondeterministic case (\ie~NAA).
% We first define the quotient for NAA.  We extend the constructions
% of~\cite{DBLP:conf/lics/LarsenX90,DBLP:journals/entcs/Raclet08} to
% nondeterministic acceptance automata.
The construction incurs an exponential blow-up, which however is local
and depends on the degree of nondeterminism.  We also provide a quotient
construction for nondeterministic MTS; this is useful because MTS
encodings for NAA can be very compact.

% In this section, we provide constructions for the adjoint to the
% structural composition, the so-called quotient.  A quotient of systems $S$
% and $T$ is a system $S\by T$ such that for each $Z$
% $$T\| Z\mr S \Leftrightarrow Z\mr S\by T$$ 

% Quotient for the much simpler case of deterministic systems has been
% constructed in~\cite{?}. The systems considered were deterministic MTS
% and deterministic BFS (also called acceptance automata). Due to the
% correspondence the construction can be carried over to the
% $\nu$-calculus preserving that for each $\zeta$ we have $\psi\|
% \zeta\rightarrow \phi \ \Rightarrow\ \zeta\rightarrow \phi\by\psi$.

% We give exponential algorithms for BFS and for MTS. Since encoding MTS
% into AA incurs an exponential blow-up it is useful to have a direct
% construction. This is also the first solution to (nondeterministic)
% MTS. We also give an example showing the exponential lower
% bound. However, the blow-up is local and depends on the
% degree of nondeterminism. In the case of deterministic systems, it is
% as efficient as the classical construction.
% We document practical usability of the algorithm by an experimental
% evaluation.

Let $(S, S^0, \Tran_S)$, $(T, T^0, \Tran_T)$ be two NAA.  We define the
quotient $S\by T=(Q,\{ q^0\}, \Tran_Q)$.  Let $Q= \Pfin{ S\times T}$
and $q^0=\{( s^0, t^0)\mid s^0\in S^0, t^0\in T^0\}$.  States in $Q$
will be written $\{ s_1\by t_1,\dots, s_n\by t_n\}$ instead
of $\{( s_1, t_1),\dots,( s_n, t_n)\}$.

In the following, we use the notation $x\in\in z$ as a shortcut
for the fact that there exists $y$ with $x\in y\in z$.
We first define $\Tran_Q(\emptyset) = 2^{\Sigma \times \{ \emptyset
  \}}$.  This means that the empty set of pairs is the universal state
$\top$.  Now let $q=\{ s_1\by t_1,\dots, s_n\by t_n\}\in Q$.  We first
define the auxiliary set of possible transitions $\PosTran(q)$ as
follows.  For $x\in S\cup T$, let $\alpha( x) = \{ a\in \Sigma\mid
\exists y: (a,y) \in\in \Tran(x) \}$ and $\gamma(q) = \bigcap_i
\big(\alpha(s_i) \cup( \Sigma\setminus \alpha(t_i))\big)$.  Let further
$\pi_a(X) = \{ x \mid (a,x) \in X \}$.  

Let now $a\in \gamma(q)$. For all $i\in\{ 1,\dots, n\}$, let $\{ t_{ i,
  1},\dots, t_{ i, m_i}\} = \pi_a(\bigcup \Tran_T(t_i))$ be the possible
next states from $t_i$ after an $a$-transition, and define
\begin{multline*}
  \PosTran_a( q)=\big\{\{ s_{ i, j}\by t_{ i, j}\mid i\in\{
  1,\dots, n\},
  j\in\{ 1,\dots, m_i\}\}\mid \\
  \forall i\in\{ 1,\dots, n\}: \forall j\in\{ 1,\dots, m_i\}: 
  (a,s_{ i, j})\in\in \Tran_S(s_i)\big\}
\end{multline*}
and $\PosTran(q) = \bigcup_{a \in \Sigma} (\{a\} \times \PosTran_a(q))$.
Hence $\PosTran_a( q)$ contains sets of possible next quotient states
after an $a$-transition, each obtained by combining the $t_{ i, j}$ with
some permutation of possible next $a$-states in $S$.
We then define
$$\Tran_Q(q) = \{ X \subseteq \PosTran(q)\mid \forall i :
\forall Y \in \Tran_T(t_i) : X\triangleright Y \in \Tran_S(s_i) \},$$
where the operator $\triangleright$ is defined by $\{s_1\by t_1, \dots,
s_k\by t_k\}\triangleright t_\ell = s_\ell$ and $X\triangleright Y = \{
(a, x\triangleright y) \mid (a,x) \in X, (a,y) \in Y \}$.  Hence
$\Tran_Q( q)$ contains all sets of (possible) transitions which are
compatible with all $t_i$ in the sense that (the projection of) their
parallel composition with any set $Y\in \Tran_T( t_i)$ is in $\Tran_S(
s_i)$.

\begin{theorem}
  \label{thm:bfs-quotient}
  For all NAA $S$, $T$ and $X$, $X \| T \mr S$ iff $X \mr S \by T$.
\end{theorem}

\begin{theorem}
  With operations $\land$, $\lor$, $\|$ and $\by$, the set of NAA forms
  a commutative residuated lattice up to $\mreq$.
\end{theorem}

This theorem makes clear the relation of NAA to linear
logic~\cite{DBLP:journals/tcs/Girard87}: except for completeness of the
lattice induced by $\land$ and $\lor$ (\cf~Theorem~\ref{th:lattice}),
NAA form a \emph{commutative unital Girard
  quantale}~\cite{DBLP:journals/jsyml/Yetter90}, the standard algebraic
setting for linear logic.  Completeness of the lattice can be obtained
by allowing infinite conjunctions and disjunctions (and infinite NAA).

\subsection{Quotient for MTS} 

We now give a quotient algorithm for the important special case of MTS,
which results in a much more compact quotient than the NAA construction
in the previous section.  However, MTS are not closed under quotient;
\cf~\cite[Thm.~5.5]{DBLP:conf/concur/Larsen90}. We show that the quotient of two
MTS will generally be a DMTS.

Let $( S, s^0, \omay_S, \omust_S)$ and $( T, t^0, \omay_T, \omust_T)$ be
nondeterministic MTS.  We define the quotient $S\by T=( Q, \{q^0\},
\omay_Q, \omust_Q)$.  We let $Q= \Pfin{ S\times T}$ as before, and
$q^0=\{( s^0, t^0)\}$.
% States in $Q$ will be written $\{ s_1\by t_1,\dots,
% s_n\by t_n\}$ instead of $\{( s_1, t_1),\dots,( s_n, t_n)\}$.
%
%The modality of actions in a
%quotient state are as follows: $\en(\{s_i \by t_i \mid i \in I\})(a) =
%\bigsqcap_{i \in I}(\en(s_{i})(a)\by\en(t_i)(a))$.
The state $\emptyset\in Q$ is again universal, so we define $\emptyset
\may{a} \emptyset$ for all $a \in \Sigma$.  There are no must
transitions from $\emptyset$.

Let $\alpha(s)$, $\gamma(q)$ be as in the previous section.
% and let $\alpha_\mathrm{must}(s) = \{ a \mid \exists t : s \must{a} t
% \}$.  We further define %a~function
% $\gamma_\mathrm{must}(q) = \bigcup_{s\by t \in q}
% \alpha_\mathrm{must}(s)$.  and a~predicate $\mathit{good}(q) =
% \bigwedge_{s\by t \in q} \alpha_\mathrm{must}(s) \subseteq
% \alpha_\mathrm{must}(t)$.  Further,
For convenience, we work with sets $\May_a(s)$, for $a\in \Sigma$ and
states $s$, instead of may transitions, \ie~we have $\May_a( s)=\{ t\mid
s\may a t\}$.

Let $q=\{ s_1\by t_1,\dots, s_n\by t_n\}\in Q$
%If $\mathit{good}(q)$ does not hold, we set $q \must{} \emptyset$ ($q$ is
%    inconsistent).
%In the following, we thus assume that $\mathit{good}(q)$ holds.
and $a\in \Sigma$. First we define the may~transitions.  If $a \in
\gamma(q)$ then for each $i\in\{ 1,\dots, n\}$, write $\May_a( t_i)=\{
t_{ i, 1},\dots, t_{ i, m_i}\}$, and define
\begin{multline*}
  \May_a( q)= \big\{\{ s_{ i, j}\by t_{ i, j}\mid i\in\{ 1,\dots, n\},
  j\in\{ 1,\dots, m_i\}\}\mid \\
  \qquad\, \forall i\in\{ 1,\dots, n\}: \forall j\in\{ 1,\dots, m_i\}: s_{ i,
    j}\in \May_a( s_i)\big\}.
\end{multline*}

For the (disjunctive) must-transitions, we let, for every
$s_i\must{a}s'$,
\begin{equation*}
  q\must{} \{(a, M) \in \{a\}\times\May_a( q)\mid \exists t': s'\by
  t'\in M,\ t_i\must a t'\}.
\end{equation*}

% The so-defined quotient $Q$ needs to be pruned.  We hence define $Q'= \rho(
% Q)$.

\begin{example}
  We illustrate the construction on an example.  Let $S$ and $T$ be the
  MTS in the left part of Fig.~\ref{fig:mtsquotient}.  We construct
  $S\by T$; the end result is displayed in the right part of the
  figure.

\begin{figure}[t]
%\figspacee
  \centering
  \begin{tikzpicture}[x=1.5cm,y=0.6cm,font=\footnotesize,
    ->,>=stealth',
    state/.style={shape=circle,draw,font=\scriptsize,inner sep=.5mm,outer
      sep=0.8mm, minimum size=0.3cm,initial text=,initial
      distance=2ex}]
    \begin{scope}
      \node[state,initial] (s) at (0,0) {$s^0$};
      \node[state] (s1) at (1,0.5) {$s_1$};
      \path (s)	edge [->] node[above]{$a$}	(s1);
      \node[state] (s2) at (1,-0.5) {$s_2$};
      \path (s)	edge [->,densely dashed] node[below]{$a$}	(s2);
      \node (end) at (2,0.5) {$\bullet$};
      \path (s1) edge [->] node[above]{$b$}	(end);
    \end{scope}
    \begin{scope}[yshift=1.3cm]
      \node[state,initial] (s) at (0,0) {$t^0$};
      \node[state] (s1) at (1,0.5) {$t_1$};
      \path (s)	edge [->] node[above]{$a$}	(s1);
      \node[state] (s2) at (1,-0.5) {$t_2$};
      \path (s)	edge [->,densely dashed] node[below]{$a$}	(s2);
      \node (end) at (2,0.5) {$\bullet$};
      \path (s1)	edge [->] node[above]{$b$}	(end);
      \node (end) at (2,-0.5) {$\bullet$};
      \path (s2)	edge [->,densely dashed] node[below]{$c$}	(end);
    \end{scope}
  \end{tikzpicture}%
%
  %\vspace{3mm}
  \begin{tikzpicture}[x=2.5cm,y=1cm,font=\footnotesize,
    ->,>=stealth',
    state/.style={shape=circle,draw,font=\scriptsize,inner sep=.5mm,outer
      sep=0.8mm, minimum size=0.5cm,initial text=,initial
      distance=2ex,rectangle,rounded corners}]
    \node[state,initial] (s) at (0,0) {$s^0\by t^0$};
    \node[state] (s1) at (1,1) {$\{s_1\by t_1,s_2\by t_2\}$};
    \path (s)	edge [->] node[above]{$a$}	(s1);
    \node[state] (s2) at (1,-1) {$\{s_2\by t_1,s_2\by t_2\}$};
    \path (s)	edge [->,densely dashed] node[below]{$a$}	(s2);
    \node[state,shape=circle] (end) at (2,0) {$\top$};
    \path (s1)	edge [->] node[above]{$b$}	(end);
    \path (s1.east)	edge [->,bend left,densely dashed]
    node[above]{$a$}	(end);
%    \node (end2) at (2,-0.5) {$\top$};
    \path (s2)	edge [->,densely dashed] node[below]{$a$}	(end);
%    \node (end3) at (.7,-1.5) {$\top$};
    \path (s)	edge [->,densely dashed] node[below]{$b,c$}	(end);
%    \path (s)	edge [->,densely dashed,bend right] node[below]{$c$}	(end3);
    % \path (end3)	edge [loop right,->,densely dashed]
    % node[right]{$\Sigma$}	(end3);
    % \path (end2)	edge [loop right,->,densely dashed]
    % node[right]{$\Sigma$}	(end2);
    \path (end)	edge [loop right,->,densely dashed]
    node[right]{$a, b, c$} (end);
  \end{tikzpicture}
  \caption{%
    \label{fig:mtsquotient}
    Two nondeterministic MTS and their quotient}
\figspacee
\end{figure}
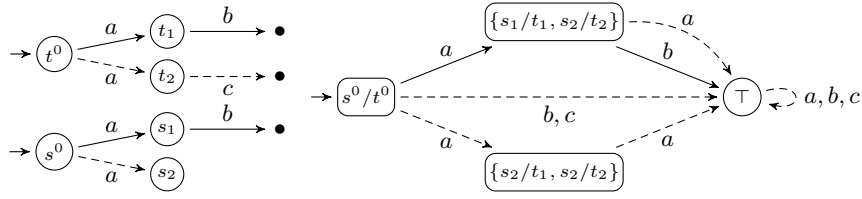

First we construct the may-successors of $s^0\by t^0$. Under $b$ and $c$
there are no constraints, hence we go to $\top$. For $a$, we have all
permutations of assignments of successors of $s$ to successors of $t$,
namely $\{s_1\by t_1,s_1\by t_2\}$, $\{s_1\by t_1,s_2\by t_2\}$,
$\{s_2\by t_1,s_1\by t_2\}$ and $\{s_2\by t_1,s_2\by t_2\}$. Since there
is a must-transition from $s$ (to $s_1$), we create a disjunctive
must-transition to all successors that can be used to yield a
must-transition when composed with the must-transition from $t$ to
$t_1$. These are all successors where $t_1$ is mapped to $s_1$, hence
the first two. However, $\{s_1\by t_1,s_1\by t_2\}$ will turn out
inconsistent, as it requires to refine $s_1$ by a composition with
$t_2$. As $t_2$ has no must under $b$, the composition has none either,
hence the must of $s_1$ can never be matched. As a result, after
pruning, the disjunctive must from $\{s^0\by t^0\}$ leads only to
$\{s_1\by t_1,s_2\by t_2\}$. Further, $\{s_2\by t_1,s_1\by t_2\}$ is
inconsistent for the same reason, so that we only have one other
may-transition under $a$ from $\{s^0\by t^0\}$.

Now $\{s_1\by t_1,s_2\by t_2\}$ is obliged to have a must under $b$ so
that it refines $s_1$ when composed with $t_1$, but cannot have any
$c$ in order to match $s_2$ when composed with $t_2$. Similarly,
$\{s_2\by t_1,s_2\by t_2\}$ has neither $c$ nor $b$.
One can easily verify that $T\|( S\by T)\mreq S$ in this case.
% The structural composition of $t$ and $\{s\by t\}$ is now exactly
% $s$. (There are three copies of the $s\may{a}\bullet$ transition
% instead of one. However, these two systems are modally equivalent.)
\end{example}

Note that the constructions may create inconsistent states, which have
no implementation. In order to get a consistent system, it needs to be
pruned. This is standard and the details can be found in
Appendix~\ref{app:pruning}. The pruning can be done in polynomial time.
%\todo{J: add complexity comment}

\begin{theorem}
  \label{thm:mts-quotient}
  For all MTS $S$, $T$ and $X$, $X\mr S\by T$ iff $T\| X\mr S$.
\end{theorem}

\myspace

\section{Conclusion and Future Work}
\label{sec:conclusion}

In this paper we have introduced a~general specification framework whose
basis consists of three different but equally expressive formalisms: one
of a~graphical behavioural kind (DMTS), one logic-based (\rhml) and one
an intermediate language between the former two (NAA).  We have shown
that the framework possesses a rich algebraic structure that includes
logical (conjunction, disjunction) and structural operations (parallel
composition and quotient). Moreover, the construction of the quotient
solves an open problem in the area of MTS. 
%% We have furthermore provided
%% procedures for determinization and a~tight approximation of the
%% complement of deterministic BFS.
%
As for future work, we hope to establish the exact complexity of the
quotient constructions. We conjecture that the exponential blow-up of
the construction is in general unavoidable.

% We have shown that Hennessy-Milner logic with maximal fixed points can
% be equipped with composition and decomposition operations which turn
% it into a complete specification theory in the sense
% of~\cite{DBLP:conf/concur/Larsen90,DBLP:conf/fase/BauerDHLLNW12}.

\bibliographystyle{plain}
\bibliography{IEEEabrv,refs}

\newpage
\appendix

\section*{Appendix: Proofs}
\section{Proofs of Section~\ref{sec:equiv}}

\begin{proof}[Proof of Proposition~\ref{pr:bfs.preorder}]
  For reflexivity of $\mr$, one only needs to see that for any NAA $S$,
  the identity relation $\id_S=\{( s, s)\mid s\in S\}\subseteq S\times
  S$ is a modal refinement from $S$ to $S$.

  To see that $\mr$ is transitive, let $S_1$, $S_2$, $S_3$ be NAA with
  $S_1\mr S_2$ and $S_2\mr S_3$.  Let $R_1$ and $R_2$ be modal
  refinement relations witnessing $S_1 \mr S_2$ and $S_2 \mr S_3$,
  respectively, and define the relation $R_3 \subseteq S_1 \times S_3$
  by $R_3=\{( s_1, s_3)\mid \exists s_2\in S_2:( s_1, s_2)\in R_1,( s_2,
  s_3)\in R_2\}$.  We show that $R_3$ is a modal refinement relation
  witnessing $S_1\mr S_3$.  Remark that as $(s^0_1, s^0_2) \in R_1$ and
  $(s^0_2, s^0_3) \in R_2$, we have $(s^0_1, s^0_3) \in R_3$.

  Let $(s_1,s_3) \in R_3$, then we have $s_2 \in S_2$ such that
  $(s_1,s_2) \in R_1$ and $(s_2,s_3) \in R_2$. Let $M_1 \in
  \Tran_1(s_1)$. By $R_1$, there exists $M_2 \in \Tran_2(s_2)$ such that
  \begin{align}
    \label{eq:mr.1}& \forall( a, t_1)\in M_1: \exists( a, t_2)\in
    M_2:( t_1, t_2)\in R_1\,, \\
    \label{eq:mr.2}& \forall( a, t_2)\in M_2: \exists( a, t_1)\in
    M_1:( t_1, t_2)\in R_1\,.
  \end{align}
  Using $R_2$, we now see that there must be $M_3 \in \Tran_3(s_3)$ for
  which
  \begin{align}
    \label{eq:mr.3}& \forall( a, t_2)\in M_2: \exists( a, t_3)\in
    M_3:( t_2, t_3)\in R_2\,, \\
    \label{eq:mr.4}& \forall( a, t_3)\in M_3: \exists( a, t_2)\in
    M_2:( t_2, t_3)\in R_2\,.
  \end{align}
  
  Now let $(a,t_1) \in M_1$. Using~\eqref{eq:mr.1}, we find $(a,t_2) \in
  M_2$ such that $(t_1,t_2) \in R_1$. By~\eqref{eq:mr.3}, there exists
  $(a,t_3) \in M_3$ such that $(t_2,t_3) \in R_2$, so that also
  $(t_1,t_3) \in R_3$.

  Conversely, let $(a,t_3) \in M_3$. By~\eqref{eq:mr.4}, there must be
  $(a,t_2) \in M_2$ such that $(t_2,t_3) \in
  R_2$. Using~\eqref{eq:mr.2}, we have $(a,t_1) \in M_1$ such that
  $(t_1,t_2) \in R_1$, and then also $(t_1,t_3) \in R_3$.

  To finish the proof, we must see that for all NAA $S$, $\bot\mr S\mr
  \top$.  The empty relation provides a witness for the former, and the
  relation $\{( s, \top)\mid s\in S\}\subseteq S\times \top$ one for the
  latter.
\end{proof}\bigskip

\begin{proof}[Proof of Theorem~\ref{th:equivalence}]
  This follows directly from Lemmas~\ref{th:bfsvsdmts},
  \ref{th:bfsvshml} and~\ref{th:hmltodmts}.
\end{proof}\bigskip

\begin{proof}[Proof of Lemma~\ref{th:bfsvsdmts}]
  The first part of the proof is trivial, as any DMTS $S$ has the same
  state-structure as its corresponding NAA $\db(S)$ and the transition
  relation in $\db(S)$ is just an enumeration of all acceptable
  choices of transitions from $S$.

  For the second part of the lemma, we need to show that for any NAA
  $S$ and any LTS $I$, $I\mr S$ (as NAA) iff $I\mr \bd(
  S)$ (as DMTS).

  Let $S = (S, S^0, \Tran)$ be a NAA and let $\bd( S) = ( T, T^0, \omay,
  \omust)$ be defined as above. Let $I = (I, \{i^0\}, \Tran_I) \equiv (I,
  \{i^0\}, \omay_I, \omust_I)$.

  \medskip \noindent {\bf $\Rightarrow$.} We first prove that $I \mr S
  \Rightarrow I \mr \bd( S)$. Assume that $I \mr S$ with witnessing modal
  refinement relation $R \subseteq I \times S$. Given $i \in I$, let
  $M_i$ be the unique set of transitions such that $\Tran_I(i) =
  \{M_i\}$. By $R$, we know that for all $(i,s) \in R$, there exists
  $M^{i,s} \in \Tran(s)$ such that
  \begin{align*}
    & \forall( a, i')\in M_i: \exists( a, t)\in
    M^{i,s}:( i', t)\in R \\
    & \forall( a, t)\in M^{i,s}: \exists( a, i')\in
    M_i:( i', t)\in R
  \end{align*}

  Given $i \in I$, we denote by $M^{i,s}$ the corresponding set in
  $\Tran(s)$, as given above.  Let $R^D \subseteq I \times T$ be the
  relation such that $(i, M) \in R^D$ iff there is $s\in S$ with $(i,s)
  \in R$ and $M = M^{i,s}$. We show that $R^D$ is a modal refinement.
  Let $(i, M^{i,s}) \in R^D$.
  \begin{itemize}
  \item Let $(a,i')$ such that $i \may{a} i'$, \ie~$i\must a i'$ as $I$
    is an implementation. By construction, we have $(a,i') \in M_i$. By
    $R$, there exists $(a,t) \in M^{i,s}$ such that $(i',t) \in R$. By
    construction of $\bd( S)$, there exists $M^{i,s}\must{} N$ such that
    $( a, M)\in N$ for all $M \in \Tran(t)$. Since $(i',t) \in R$, there
    exists $M^{i',t} \in \Tran(t)$ such that $(i', M^{i',t})\in
    R^D$. Thus, we have $s \may{a} M^{i',t}$ and $(i', M^{i',t})\in
    R^D$.

  \item Let $M^{i,s})\must{} N$. By construction of $\bd( S)$, $N$ is of
    the form $\{(a, M)\mid M \in \Tran(t)\}$ for some $(a,t) \in
    M^{i,s}$. By $R$, there thus exists $i \must{a} i'$ such that
    $(i',t) \in R$. As a consequence, we have $(i', M^{i',t})\in R^D$
    and $(a, M^{i',t})\in N$.
  \end{itemize}

  We have shown that $R^D$ is a modal refinement relation (for
  DMTS).  We proceed to prove that it is initialised.  %% Let $i^0\in I^0$,
  %%then 
  We have $s^0\in S^0$ with $( i^0, s^0)\in R$.  By definition of
  $R^D$, this implies that $( i^0, M^{ i^0, s^0})\in R^D$, but $M^{ i^0,
    s^0}\in \Tran( s^0)$, hence $M^{ i^0, s^0}\in T^0$.
  % Moreover, we have by construction that $(i^0,(s^0, M^{i^0,s^0})) \in
  % R^D$ and $(s^0, M^{i^0,s^0}) \in T^0$, so $I \mr \bd( S)$.

  \medskip \noindent {\bf $\Leftarrow$.} We now prove that $I \mr \bd( S)
  \Rightarrow I \mr S$. Assume that $I \mr \bd( S)$ with witnessing modal
  refinement relation $R^D \subseteq I \times T$. Given $i
  \in I$, let $M_i$ be the unique set of transitions such that
  $\Tran_I(i) = \{M_i\}$. Let $R \subseteq I \times S$ be the relation
  such that $(i,s) \in R$ iff there exists $M\in \Tran(s)$ such that
  $(i, M) \in R^D$. We show that $R$ is a modal refinement.

  Let $(i,s) \in R$ and let $M^{i,s} \in \Tran(s)$ be such that
  $(i, M^{i,s}) \in R^D$.

  \begin{itemize}
  \item Let $(a,i') \in M_i$. By construction, we have $i \may{a} i'$,
    so by $R^D$, there exists $M \in T$ such that $M^{i,s}\may{a} M$. By
    construction of $\bd( S)$, there must exist $M^{i,s}\must{} N$ with
    $(a, M)\in N$. As a consequence, again by construction of $\bd( S)$,
    we must have $t\in S$ with $( a, t)\in M^{i,s}$ and $M\in
    \Tran(t)$. Therefore, there exists $(a,t) \in M^{i,s}$ such that
    $(i',t)\in R$.
  \item Let $(a,t) \in M^{i,s}$. By construction, we have
    $M^{i,s}\must{} N$ with $N=\{( a,M)\mid M\in \Tran(t)\}$. By $R^D$,
    there exists $i \must{a} i'$ such that $(i', M)\in R^D$ for some
    $M$. As a consequence, there exists $(a,i')\in M_i$ such that
    $(i',t) \in R$.
  \end{itemize}

  Hence $R$ is a modal refinement relation (for NAA).  To show that $R$
  is initialised, %% let $i^0\in I^0$, then 
  we have $N^0\in T^0$ with $(
  i^0, N^0)\in R^D$.  But then $N^0\in \Tran( s^0)$ for some $s^0\in
  S^0$, and by definition of $R$, $( i^0, s^0)\in R$.
  % Moreover, since $I \mr \bd( S)$, there exists $p^0 = (s^0,M^0) \in
  % T^0$ such that $(i^0,(s^0,M^0)) \in R^D$. Therefore, $(i^0,s^0) \in
  % R$ and $I \mr S$.
\end{proof}\bigskip

\begin{proof}[Proof of Lemma~\ref{th:bfsvshml}]
  % For the translation from BFS to HML declarations, we note that
  % $\Delta( s)$ is just the characteristic formula of $s$ as given
  % in~\cite{DBLP:conf/avmfss/Larsen89}.  This direction thus follows
  % from~\cite[Thm.~4.1]{DBLP:conf/avmfss/Larsen89}.
%%   To show that $S\treq \bh( S)$ for all BFS $S$, 
  Let $( S, S^0, \Tran)$ be a~NAA 
%%   in which all formulae $\Phi( s)$ are in the form presupposed
%%   in the construction, 
  and write $\bh( S)=( S, S^0, \Delta)$.  Let $( I,
  i^0, \omust)$, with $\omust\subseteq I\times \Sigma\times I$, be an
  LTS; we need to show that $I\in \impl S$ iff $I\in \impl \Delta$.
 
  For states $i\in I$, $s\in S$, write $i\mr s$ iff $( I, i, \omust)\in
  \impl{( S,\{ s\}, \Tran)}$, \ie~if the LTS $I$ with its initial state
  replaced by $i$ implements the BFS $S$ with initial state $s$.
  Similarly, write $i\models s$ iff $( I, i, \omust)\in \impl{( S,\{ s\},
    \Delta)}$.  
  We show that 
%% $i\mr s$ iff $i\models s$ for all $i\in I$,
%%  $s\in S$; this will imply that 
  $I\in \impl S$ iff $I\in \impl \Delta$.

  We start with the only if part. The proof is done by coinduction.  
  We define the assignment $\sigma : S \to 2^I$ as follows:
  $\sigma(t) = \{j \in I \mid j \mr t\}$. 
  We need to show that for every $s \in S$,
  $\sigma(s) \subseteq \sem{\Delta(x)}\sigma$.
  Let $i \in \sigma(s)$.
  As $i \mr s$, we know that there exists $M \in \Tran(s)$ satisfying
  the conditions of modal refinement.
  For every $(a,t) \in M$ there thus exists $i \must{a} j$
  such that $j \mr t$. This means that $j \in \sigma(t)$ and
  $i \in \sem{\langle a\rangle t}\sigma$. As $(a,t) \in M$ is arbitrary,
  this also means that $i \in \sem{\bigland_{(a,t)\in M} \langle a\rangle t}
  \sigma$.
  Let now $a \in \Sigma$ be arbitrary.
  Due to the first condition of modal refinement, we know that
  for every $i \must{a} j$ there has to be at least one $(a,u) \in M$ 
  (\ie $u \in M_a$) such that $j \mr u$. This means that 
  for every such $j$, 
  $j \in \sigma(u) \subseteq \sem{\biglor_{u \in M_a} u}\sigma$
  and thus $i \in \sem{[a]\left(\biglor_{u \in M_a} u\right)}\sigma$.
  As $a$ was arbitrary, this means that 
  $i \in \sem{\bigland_{a\in\Sigma} [a]\left(\biglor_{u \in M_a} u\right)}
  \sigma$. Together with the previous observation, we have
  $i \in  \sem{\bigland_{(a,t)\in M} \langle a\rangle t \wedge
  \bigland_{a\in\Sigma} [a]\left(\biglor_{u \in M_a} u\right)} 
  \subseteq \sem{\Delta}\sigma$.
  Clearly, 
%% for every $i^0 \in I^0$ 
  there is $s^0 \in S^0$ such that
  $i^0 \in \sigma(s^0)$. Therefore, $I \models \Delta$.

%%   by
%%   structural induction on $\Phi( s)$.  The proof is trivial for the
%%   cases $\Phi( s)= \ttt$, $\Phi( s)= \ff$, $\Phi( s)= \phi_1\land
%%   \phi_2$ and $\Phi( s)= \phi_1\lor \phi_2$, so assume $\Phi( s)=( a,
%%   t)$.  Then $i\must a j$ with $j\mr t$.  By induction hypothesis, and
%%   because we are using the \emph{maximal} fixed-point semantics of HML,
%%   we can assume that $j\models t$, hence $i\models s$.
%% 
%%   For the case $\Phi( s)= \neg( a, t)$, assume first that $i\not\must
%%   a$, then $i\models[ a] u$ for any $u\in S$, thus $i\models s$.  Now
%%   let $i\must a j$ be any $a$ transition from $i$, then $i\mr s$ implies
%%   that $j\mr u$ for some $u\in S$.  By induction hypothesis, $j\models
%%   u$, and $u\ne t$ because $\Phi( s)= \neg( a, t)$.  We conclude that
%%   $i\models s$.

  We now show the if part. We define a~relation $R$ as follows:
  \[ R = \{ (j,t) \mid j \in I, t \in S, j \models t \} \]
  and show that $R$ satisfies the conditions of modal refinement.
  Let $(i,s) \in R$. As $i \models s$ there has to exist some 
  $M \in \Tran(s)$ such that 
  $i \models \bigland_{(a,t)\in M} \langle a\rangle t \wedge
  \bigland_{a\in\Sigma} [a]\left(\biglor_{u \in M_a} u\right)$.
  Let $i \must{a} j$. As 
  $i \models [a]\left(\biglor_{u \in M_a} u\right)$, there has 
  to be some $(a,u) \in M$ such that $j \models u$. The first condition
  of modal refinement is thus met.
  Let further $(a,t) \in M$. As $i \models \langle a\rangle t$, 
  this means that there is some $i \must{a} j$ such that 
  $j \models t$. The second condition of modal refinement is thus
  also met.
  Clearly, $R$ also satisfies the condition of an initialised refinement.
  Thus $I \mr S$. 

\end{proof}\bigskip

\begin{proof}[Proof of Lemma~\ref{le:hmlnormalstrong}]
  It is shown in~\cite{DBLP:journals/tcs/BoudolL92} that any HML formula
  is equivalent to one in \emph{strong normal form}, \ie~either $\ttt$
  or of the form $\biglor_{ i\in I}( \bigland_{ j\in J_i} \langle a_{
    ij}\rangle \phi_{ ij}\land \bigland_{ a\in \Sigma}[ a] \psi_{ i,
    a})$ for HML formulas $\phi_{ ij}$, $\psi_{ i, a}$ which are also in
  strong normal form.  We only need to replace the $\phi_{ ij}$, $\psi_{
    i, a}$ by (new) variables $x_{ ij}$, $y_{ i, a}$ and add
  declarations $\Delta_2( x_{ ij})= \phi_{ ij}$, $\Delta_2( y_{ i, a})=
  \psi_{ i, a}$ to finish the proof.
\end{proof}\bigskip

\begin{proof}[Proof of Lemma~\ref{th:hmltodmts}]
  Let $( x, k)\in S$, with $\Delta( x)=\biglor_{ i\in I}( \bigland_{
    j\in J_i} \langle a_{ ij}\rangle x_{ ij}\land \bigland_{ a\in
    \Sigma}[ a] y_{ i, a})$ and $I\ne \emptyset$.  By construction, the
  characteristic formula~\cite{DBLP:conf/avmfss/Larsen89} of $( x, k)$
  is $\chi( x, k)= \bigland_{ j\in J_i} \langle a_{ ij}\rangle(
  \biglor_{ i\in I_j} x_{ ij})\land \bigland_{ a\in \Sigma}[ a](
  \biglor_{ i\in I_a} y_{ i, a})$.  Distributing the disjunctions over
  the conjunctions, we see that $\Delta( x)= \biglor_k \chi( x, k)$.

  Now let $( I, i^0, \omust)$ be a LTS.  Then
  \begin{align*}
    I\models \Delta &\Leftrightarrow \exists x^0\in
    X^0: i^0\models \Delta( x^0) \\
    &\Leftrightarrow \exists x^0\in X^0: \exists k:
    i^0\models \chi( x^0, k) \\
    &\Leftrightarrow \exists( x^0, k)\in S^0:
    i^0\mr( x^0, k) \\
    &\Leftrightarrow I\mr S,
  \end{align*}
  the next-to-last biimplication holds precisely because $\chi( x^0, k)$ is
  the characteristic formula of $( x, k)$.
\end{proof}\bigskip

\section{Proofs of Section~\ref{sec:lattice}}

\begin{proof}[Proof of Theorem~\ref{th:disj-impl}]
  Let $S_1$ and $S_2$ be DMTS or NAA. Let $I$ be an implementation
  such that $I \in \impl{S_1 \lor S_2}$, i.e. $I \mr S_1 \lor
  S_2$. Let $R$ be the initialised modal refinement witnessing $I \mr
  S_1 \lor S_2$. By construction of $S_1 \lor S_2$, $R$ can be split
  into two relations $R_1 = R \cap S_1$ and $R_2 = R \cap S_2$ such
  that $R = R_1 \cup R_2$. One can then verify that both $R_1$ and
  $R_2$ are modal refinement relations. Depending on the equivalence
  class of the initial state of $I$ (either belonging to $R_1$ or
  $R_2$), one can verify that either $I \mr S_1$ or $I \mr S_2$. As a
  consequence, $I \in \impl{S_1} \cup \impl{S_2}$, thus $\impl{S_1
    \lor S_2} \subseteq \impl{S_1} \cup \impl{S_2}$.
  Conversely, if $I \mr S_1$ (resp $S_2$) with modal refinement
  relation $R$, one can verify that
  $R$ also witnesses $I \mr S_1 \lor S_2$. Thus $\impl{S_1} \cup
  \impl{S_2} \subseteq \impl{S_1 \lor S_2}$.
\end{proof}\bigskip

\begin{proof}[Proof of Lemma~\ref{le:bfs.oneinitial}]
  Write $S=( S, S^0, \Tran)$.  If $S^0= \emptyset$, we can let $T=\{
  t^0\}$ and $\Tran( t^0)= \emptyset$; note that $S\mreq T\mreq \bot$.
  Otherwise, we let $T= S\cup\{ t^0\}$, where $t^0$ is a new state, and
  $\Tran( t^0)= \bigcup_{ s^0\in S^0} \Tran( s^0)$.  Let $R=\id_S\cup\{(
  s^0, t^0)\mid s^0\in S^0\}$, then $R$ is an initialised refinement
  $S\mr T$ and the inverse relation $R^{ -1}$ an initialised refinement
  $T\mr S$.
\end{proof}\bigskip

\begin{proof}[Proof of Lemma~\ref{le:conj-dmts}]
  Let $S_1 = (S_1,S^0_1,\omay_1,\omust_1)$ and $S_2 =
  (S_2,S^0_2,\omay_2,\linebreak[4]\omust_2)$ be DMTS. Let $db(S_1) = (S_1, S^0_1,
  \Tran_1)$ and $\db(S_2) = (S_2, S^0_2, \Tran_2)$ be
  their corresponding NAA. Let $S^{\land} = \db( S_1\land S_2)$ and
  $S_{\land} = \db( S_1)\land \db( S_2)$.  We show that $S^{\land}$
  and $S_{\land}$ are syntactically equivalent.

  First, remark that $S^{\land}$ and $S_{\land}$ have precisely the
  same state-space, which is $S_1 \times S_2$, and initial states,
  which are $S_1^0 \times S_2^0$. We now show that they have the same
  transition functions. Let $\Tran_{\land}$ (resp. $\Tran^{\land}$) be
  the transition function of $S_{\land}$ (resp. $S^{\land}$). 
%% Let
%%  $\Phi_{\land}$ and $\Phi^{\land}$ be the corresponding Boolean
%%  formula. 
  Let $(s_1,s_2) \in S_1 \times S_2$ and let $M \subseteq
  \Sigma \times S_1 \times S_2$ be such that $M \in
  \Tran_{\land}((s_1,s_2))$.

  By construction of $\Tran_{\land}$, there must be $M_1 \in
  \Tran_1(s_1)$ and $M_2 \in \Tran_2(s_2)$ such that $M \in M_1 \land
  M_2$, i.e. $\pi_1(M) = M_1$ and $\pi_2(M) = M_2$. We show that $M
  \in \Tran^{\land}((s_1,s_2))$.
%% , i.e. that $M$ satisfies $\Phi^{\land}((s_1,s_2))$.

  \begin{itemize}
  \item Let $(a,(t_1,t_2)) \in M)$. Since $\pi_1(M) = M_1$ and
    $\pi_2(M) = M_2$, we have $(a,t_1) \in M_1$ and $(a,t_2) \in
    M_2$. As a consequence, there are transitions $s_1 \may{a} t_1$
    and $s_2 \may{a} t_2$ in $S_1$ and $S_2$ respectively. Thus, by
    construction of conjunction of DMTS, there is a transition
    $(s_1,s_2)\may{a}(t_1,t_2)$ in $S_1 \land S_2$.

  \item Let $N \subseteq \Sigma \times S_1 \times S_2$ such that
    $(s_1,s_2) \omust N$ in $S_1 \land S_2$. By construction, $N$ is
    such that either (1) there exists $N_1$ such that $s_1 \omust N_1$
    in $S_1$ and $N = \{( a,( t_1, t_2))\mid( a, t_1)\in N_1,( s_1,
    s_2)\may a( t_1, t_2)\}$, or (2) there exists $N_2$ such that $s_2
    \omust N_2$ in $S_2$ and $N = \{( a,( t_1, t_2))\mid( a, t_2)\in
    N_2,( s_1, s_2)\may a( t_1, t_2)\}$. Assume that (1) holds (case (2)
    being symmetric). Since $M_1 \in \Tran_1(s_1)$, there
    must be $(a,t_1) \in N_1 \cap M_1$. Since $\pi_1(M) = M_1$, there
    must be $t_2 \in S_2$ such that $(a,(t_1,t_2)) \in M$. As a
    consequence, there is $(a,(t_1,t_2)) \in M \cap N$.
  \end{itemize}

  Finally, %% $M$ satisfies $\Phi^{\land}((s_1,s_2))$ and thus 
  $M \in \Tran^{\land}((s_1,s_2))$.

  Conversely, we can show that for all $M \in
  \Tran^{\land}((s_1,s_2))$, we also have $M \in
  \Tran_{\land}((s_1,s_2))$ in a similar way. We can thus conclude that
  $\Tran^{\land} = \Tran_{\land}$ and thus that $S^{\land}$ and
  $S_{\land}$ are syntactically equivalent.
\end{proof}\bigskip

To prove \emph{Theorem~\ref{co:conj=glb}}, we need the following lemma:

\begin{lemma}
  \label{th:conj=glb}
  For NAA or DMTS $S_1$, $S_2$, $S_3$, $S_1 \mr S_2 \wedge S_3$ iff $S_1
  \mr S_2$ and $S_1 \mr S_3$.
\end{lemma}
\begin{proof}%[Proof of Lemma~\ref{th:conj=glb}]
  We prove the two implications separately.

  \noindent {$\mathbf \Leftarrow.$} Let $S_1, S_2, S_3$ be NAA with $S_i
  = (S_i, s^0_i, \Tran_i)$ and consider the conjunction $S_2 \wedge S_3
  = (S,s^0, \Tran)$. Assume that $S_1 \mr S_2$ with witnessing relation
  $R_2 \subseteq S_1 \times S_2$ and that $S_1 \mr S_3$ with witnessing
  relation $R_3 \subseteq S_1 \times S_3$. We prove that $S_1 \mr (S_2
  \wedge S_3)$. Consider the relation $R \subseteq S_1 \times (S_2
  \times S_3)$ such that $(s_1,(s_2,s_3)) \in R \iff (s_1,s_2) \in R_2
  \land (s_1,s_3) \in R_3$. We prove that $R$ is a modal refinement. Let
  $(s_1,(s_2,s_3)) \in R$ and $M_1 \in \Tran_1(s_1)$. By $R_2$, there
  exists $M_2 \in \Tran_2(s_2)$ such that
  \begin{align}
    & \forall( a, t_1)\in M_1: \exists( a, t_2)\in
    M_2:( t_1, t_2)\in R_2 \label{proof:conj:eq1}\\
    & \forall( a, t_2)\in M_2: \exists( a, t_1)\in M_1:( t_1, t_2)\in
    R_2. \label{proof:conj:eq2}
  \end{align}
  Moreover, by $R_3$, there exists $M_3 \in \Tran_3(s_3)$ such that
  \begin{align}
    & \forall( a, t_1)\in M_1: \exists( a, t_3)\in
    M_3:( t_1, t_3)\in R_3 \label{proof:conj:eq3}\\
    & \forall( a, t_3)\in M_3: \exists( a, t_1)\in M_1:( t_1, t_3)\in
    R_3. \label{proof:conj:eq4}
  \end{align}

  % As a consequence, by combining (\ref{proof:conj:eq1}),
  % (\ref{proof:conj:eq2}), (\ref{proof:conj:eq3}) and
  % (\ref{proof:conj:eq4}), we have that $\Act(M_1) = \Act(M_2) =
  % \Act(M_3)$, thus $M_2 \wedge M_3 \ne \emptyset$.
  We construct the set $M$ using the following principle: for all
  $(a,t_2) \in M_2$, we know by (\ref{proof:conj:eq2}) that there exists
  $(a,t_1) \in M_1$ such that $(t_1,t_2) \in R_2$.  Given the state
  $t_1$, we know by (\ref{proof:conj:eq3}) that there exists $(a,t_3)
  \in M_3$ such that $(t_1,t_3) \in R_3$. The set $M$ is thus composed
  of the transitions obtained by combining (\ref{proof:conj:eq2}) and
  (\ref{proof:conj:eq3}) and (\ref{proof:conj:eq1}) and
  (\ref{proof:conj:eq4}):
  \begin{multline*}
    M = \{(a,(t_2,t_3))\mid (a,t_2)\in M_2, (a,t_3)\in M_3: \\
     \exists (a,t_1) \in M_1, (t_1,t_2)\in R_2 \land (t_1,t_3) \in
    R_3\}.
  \end{multline*}

  % Let $M$ be the maximal element in $M_2 \wedge M_3$, i.e. $M =
  % \{(a,(t_2,t_3)) \mid (a,t_2) \in M_2 \land (a,t_3)\in M_3\}$.
  By construction, we know that $M \in \Tran(s_2,s_3)$.
  \begin{itemize}
  \item Let $(a,t_1) \in M_1$. Consider states $t_2$ and $t_3$ given by
    (\ref{proof:conj:eq1}) and (\ref{proof:conj:eq3})
    respectively. Since $(a,t_2) \in M_2$, $(a,t_3) \in M_3$, $(t_1,t_2)
    \in R_2$ and $(t_1,t_3) \in R_3$ we have $(a,(t_2,t_3)) \in M$ and
    $(t_1,(t_2,t_3)) \in R$.
  \item Let $(a,(t_2,t_3)) \in M$. By construction of $M$, there exists
    $(a,t_2) \in M_2$, $(a,t_3) \in M_3$ and $(a,t_1) \in M_1$ such that
    $(t_1,t_2) \in R_2$ and $(t_1,t_3) \in R_3$, thus $(t_1, (t_2,t_3))
    \in R$.
  \end{itemize}
  By construction, we know that $(s^0_1, (s^0_2,s^0_3)) \in R$, thus $R$
  is a modal refinement relation and $S_1 \mr (S_2 \wedge S_3)$.

  \medskip
  \noindent {$\mathbf \Rightarrow.$} Let $S_1, S_2, S_3$ be NAA with
  $S_i = (S_i, s^0_i, \Tran_i)$ and consider the conjunction $S_2 \wedge
  S_3 = (S,s^0, \Tran)$. Assume that $S_1 \mr (S_2 \wedge S_3)$ with a
  witnessing relation $R$. We show that $S_1 \mr S_2$ ($S_1 \mr S_3$ is
  then obtained by symmetry). Let $R_2 \subseteq S_1 \times S_2$ be the
  relation such that $(s_1, s_2) \in R_2 \iff \exists s_3 \in S_3$
  s.t. $(s_1,(s_2,s_3)) \in R$. We show that $R_2$ is a modal refinement
  relation. Let $(s_1,s_2) \in R_2$ and consider $s_3 \in S_3$ such that
  $(s_1,(s_2,s_3)) \in R$. Let $M_1 \in \Tran_1(s_1)$. By $R$, we know
  that there exists $M \in \Tran((s_2,s_3))$ such that
  \begin{align}
    & \forall( a, t_1)\in M_1: \exists( a, (t_2,t_3))\in
    M:( t_1, (t_2,t_3))\in R \label{proof:conj:eq5}\\
    & \forall( a, (t_2,t_3))\in M: \exists( a, t_1)\in M_1:( t_1,
    (t_2,t_3))\in R. \label{proof:conj:eq6}
  \end{align}

  Consider $M_2 = \pi_2( M)$. By construction of $\Tran((s_2,s_3))$, we
  know that $M_2 \in \Tran_2(s_2)$.
  \begin{itemize}
  \item Let $(a,t_1) \in M_1$. By (\ref{proof:conj:eq5}), there exists
    $(a,(t_2,t_3)) \in M$ such that $(t_1,(t_2,t_3)) \in R$. As a
    consequence, we have $(a,t_2) \in M_2 = \pi_2( M)$ and $(t_1,t_2) \in
    R_2$.
  \item Let $(a,t_2) \in M_2$. By construction, there exists $t_3 \in
    S_3$ such that $(a,(t_2,t_3)) \in M$. By (\ref{proof:conj:eq6}),
    there exists $(a,t_1) \in M_1$ such that $(t_1,(t_2,t_3)) \in R$,
    thus $(t_1,t_2) \in R_2$.
  \end{itemize}

  Finally, we know that $(s^0_1,(s^0_2,s^0_3)) \in R$, thus
  $(s^0_1,s^0_2) \in R_2$ and $R_2$ is a modal refinement relation such
  that $S_1 \mr S_2$.
\end{proof}\bigskip

\begin{proof}[Proof of Theorem~\ref{co:conj=glb}]
  The result directly follows from Lemma~\ref{th:conj=glb}. 
  Let $S_1$ and $S_2$ be NAA or DMTS. Let $I \in \impl{ S_1\land
    S_2}$, we thus have $I \mr S_1 \land S_2$. By
  Lemma~\ref{th:conj=glb}, we thus have $I \mr S_1$ and $I \mr S_2$,
  thus $I \in \impl{S_1} \cap \impl{S_2}$. Reversely, if $I \in
  \impl{S_1} \cap \impl{S_2}$, then we have $I \mr S_1$ and $I \mr
  S_2$. By Lemma~\ref{th:conj=glb}, this implies that $I \mr S_1
  \land S_2$, and thus $I \in \impl{S_1 \land S_2}$.
\end{proof}\bigskip

\begin{proof}[Proof of Theorem~\ref{th:lattice}]
  The sets form bounded lattices by standard order-theoretic arguments,
  so only the distributive law remains to be verified.  Let thus $S_1$,
  $S_2$, $S_3$ be DMTS (the argument for NAA is similar); we want to
  show that $S_1\land( S_2\lor S_3)\mreq( S_1\land S_2)\lor( S_1\land
  S_3)$.  The state spaces of both sides are $S_1\times S_2\cup
  S_1\times S_3$, and it is easily verified that the identity relation
  is a two-sided modal refinement.
\end{proof}\bigskip

\begin{proof}[Proof of Lemma~\ref{le:bfs||prop}]
  Associativity and commutativity are clear.  To show distributivity
  over $\lor$, let $S_1$, $S_2$, $S_3$ be NAA. We prove that $S_1\|(
  S_1\lor S_3)\mreq S_1\| S_2\lor S_1\| S_3$; right-distributivity will
  follow by commutativity.  The state spaces of both sides are
  $S_1\times S_2\cup S_1\times S_3$, and it is easily verified that the
  identity relation is a two-sided modal refinement.

  For the claim that $S\| \mathsf{U}\mreq S$ for all NAA $S$, let $u$ be
  the unique state of $\mathsf{U}$ and define $R=\{(( s, u), s)\mid s\in
  S\}\subseteq S\times \mathsf{U} \times S$.  We show that $R$ is a
  two-sided modal refinement.  Let $(( s, u), s)\in R$ and $M\in \Tran(
  s, u)$, then there must be $M_1\in \Tran( s)$ for which $M= M_1\|(
  \Sigma\times\{ u\})$.  Thus $M_1=\{( a, t)\mid( a,( t, u))\in M\}$.
  Then any element of $M$ has a corresponding one in $M_1$, and vice
  versa, and their states are related by $R$.

  For the other direction, let $M_1\in \Tran( s)$, then $M= M_1\|(
  \Sigma\times\{ u\})=\{( a,( t, u))\mid( a, t)\in M_1\}\in \Tran( s,
  u)$, and the same argument applies.
\end{proof}\bigskip

\begin{proof}[Proof of Theorem~\ref{th:bfs||indimp}]
  Let $S_1\mr S_3$ and $S_2\mr S_4$, then $S_1\lor S_3\mreq S_3$ and
  $S_2\lor S_4\mreq S_4$.  By distributivity, $S_3\| S_4\mreq( S_1\lor
  S_3)\|( S_2\lor S_4)\mreq S_1\| S_2\lor S_1\| S_3\lor S_3\| S_2\lor
  S_3\| S_4$, thus $S_1\| S_2\lor S_1\| S_3\lor S_3\| S_2\mr S_3\|
  S_4$.  But $S_1\| S_2\mr S_1\| S_2\lor S_1\| S_3\lor S_3\| S_2$,
  finishing the argument.
\end{proof}\bigskip

\subsection{Proof of the NAA Quotient---Theorem~\ref{thm:bfs-quotient}}

We assume that for each $t \in T$ the elements of $\Tran_T(t)$ are
pairwise disjoint.  This assumption can easily be enforced by expanding
the state space: if $M_1, M_2\in \Tran( t)$ with $( a, u)\in M_1$ and $(
a, u)\in M_2$, we can replace the second occurrence by $( a, u')\in
M_2$, where $u'$ is a new state with $\Tran( u')= \Tran( u)$.

\begin{lemma}\label{lem:bq0}
For all $j$, $\{s_1 \by t_1, \ldots, s_n \by t_n \} \mr \{ s_j \by t_j \}$.
\end{lemma}
\begin{proof}
We show that $\supseteq$ restricted to elements of $Q$ is a~modal refinement relation, which is straightforward.
\end{proof}

\begin{lemma}\label{lem:bq1}
$ X \mr S \by T \Rightarrow X \parallel T \mr S $
\end{lemma}
\begin{proof}
Assume that $X \mr S \by T$.
We let
\[ \mathsf R = \{ (x \parallel t, s) \mid x \mr s \by t \} \]
and show that $\mathsf R$ is a~modal refinement relation.

Let now $(x \parallel t, s) \in \mathsf R$ and let $M \in \Tran(x\parallel t)$.
This means that $M = M_1 \parallel M_2$ where $M_1 \in \Tran(x)$ and 
$M_2 \in \Tran(t)$.
As we know that $x \mr s\by t$, for $M_1$ there has to exist corresponding
$N \in \Tran(s\by t)$ satisfying the conditions of modal refinement (*).
Let now $N' = combine(N,M_2)$ correspond to $M$. We prove the two conditions:
\begin{itemize}
\item Let $(a,x'\parallel t') \in M$. Then $(a,x') \in M_1$ and $(a,t') \in M_2$.
Due to (*) there has to exist $(a,q) \in N$ with $x' \mr q$ 
where $q = \{ s_1 \by t_1, \ldots,
s_k \by t_k \}$. Due to the construction of the quotient, there has to
be some $j$ such that $t' = t_j$. Therefore, $(a,s_j) \in N'$.
Due to Lemma~\ref{lem:bq0}, $x' \mr q$ implies $x' \mr s_j \by t_j$
and thus $(x'\parallel t_j,s_j) \in \mathsf R$.
\item Let $(a,s') \in N'$. This means that $(a,q) \in N$
with $s'\by t' \in q$
and $(a,t') \in M_2$. Due to (*) there has to exist $(a,x') \in M_1$ with
$x' \mr q$. Therefore, $(a,x'\parallel t') \in M_1 \parallel M_2 = M$.
Again, $x' \mr q$ implies $x' \mr s'\by t'$ and thus
$(x'\parallel t_j,s_j) \in \mathsf R$.
\end{itemize}
Obviously, as $x_0 \mr s_0 \by t_0$, we have 
$(x_0 \parallel t_0, s_0) \in \mathsf R$.
Therefore $X \parallel T \mr S$.
\end{proof}

%\begin{lemma}\label{lem:bq2}
%If there exists $X \parallel T \mr S$, then $S \by T$ exists.
%\end{lemma}
%\begin{proof}
%If $q = \{ s_1 \by t_1, \ldots, s_n \by t_n \}$ is pruned
%then either $\Tran(q) = \emptyset$ or every element of $\Tran(q)$ contains
%some $(a,q')$ where $q'$ is pruned. We show that $q$ is pruned implies
%there can be no $x$ such that $x \parallel t_i \mr s_i$ for all $i$.
%
%If $\Tran(q) = \emptyset$ \ldots
%
%TBD
%\end{proof}

\begin{lemma}\label{lem:bq3}
$ X \parallel T \mr S \Rightarrow X \mr S \by T $
\end{lemma}

\begin{proof}
Assume that $X \parallel T \mr S$.
We let
\[ \mathsf R = \{ (x, \{ s_1 \by t_1, \ldots, s_k \by t_k \}
	\mid \forall j : x \parallel t_j \mr s_j \} \]
and show that $\mathsf R$ is a~modal refinement relation.

Let now $(x,q =  \{ s_1 \by t_1, \ldots, s_k \by t_k \})
\in \mathsf R$ and let $M \in \Tran(x)$. We show how to build a~corresponding
$N \in \Tran(q)$.

For every $j$, let $\Tran(t_j) = \{ N_{j,1}, \ldots, N_{j,m_j} \}$.
As these are pairwise disjoint, every $(a,t') \in\in \Tran(t_j)$ 
may be assigned its $N_{j,\ell}$, we denote this as $\delta(a,t')$.
Let $M_{j,\ell} = M \parallel N_{j,\ell} \in \Tran(x \parallel t_j)$.
As $x \parallel t_j \mr s_j$ this means that for $M_{j,\ell}$ there is
a~corresponding $K_{j,\ell} \in \Tran(s_j)$ satisfying the conditions of
modal refinement (*1) and (*2).

Let now for every $a \in \alpha(x)$
\begin{multline*}
N_a = \{ r \in PosTran_a(q) \mid 
\exists (a,\bar x) \in M : 
\forall \bar s \by \bar t \in r: \\
	\delta(a,\bar t) = N_{j,\ell} : (a,\bar s) \in K_{j,\ell} \text{ and }
%	\exists (a,\bar x\parallel \bar t) \in M_{j,\ell} : 
	\bar x\parallel \bar t \mr \bar s 
\}
\end{multline*}
\[ N = \bigcup_{a\in\alpha(x)} \{a\} \times N_a \]
We need to show that $N \in \Tran(q)$.
Let $j$ be arbitrary and let $N_{j,\ell} \in \Tran(q)$.
We claim that $combine(N,N_{j,\ell}) = K_{j,\ell}$.
Obviously, the $\subseteq$ part holds, so we only prove $\supseteq$.
Let thus $(a,s') \in K_{j,\ell}$. Due to (*2) there has to exist
$(a,x' \parallel t') \in M_{j,\ell}$ such that $x'\parallel t' \mr s'$.
But then also $\delta(a,t') = N_{j,\ell}$ and there exists $(a,r) \in N$
such that $r$ contains $s' \by t'$. Therefore 
$(a,s') \in combine(N,N_{j,\ell})$.

\begin{itemize}
\item Let $(a,x') \in M$. 
For every $j$ and every $(a,t') \in\in \Tran(t_j)$
let $\delta(a,t') = N_{j,\ell}$ and choose $(a,s') \in K_{j,\ell}$ such that
$x' \parallel t' \mr s'$. Such $s'$ has to exist due to (*1).
Denote this by $chosen(t') = s'$.
The set $r = \{ s' \by t' \mid (a,t') \in\in \Tran(t_j), chosen(t') = s'\}$
is in $N_a$. Therefore $(a,r) \in N$ and clearly $(x',r) \in \mathsf R$.
\item Let $(a,r) \in N$. This means that $r \in N_a$
and due to the definition of $N_a$ there has to exist $(a,\bar x) \in M$
satisfying certain conditions,
notably that for all $\bar s \by \bar t \in r$ we have 
$x \parallel \bar t \mr \bar s$. This means that $(\bar x, r) \in \mathsf R$.
\end{itemize}
Obviously, as $x_0 \parallel t_0 \mr s_0$, we have $(x_0, \{ s_0 \by t_0 \})
\in \mathsf R$. Therefore $X \mr S \by T$.
\end{proof}\bigskip

\subsection{Proof of the MTS Quotient---Theorem~\ref{thm:mts-quotient}}

\begin{lemma}\label{lem:mq0}
For all $j$, $\{s_1 \by t_1, \ldots, s_n \by t_n \} \mr \{ s_j \by t_j \}$.
\end{lemma}
\begin{proof}
We show that $\supseteq$ restricted to elements of $Q$ is a~modal refinement relation, which is straightforward.
\end{proof}

\begin{lemma}\label{lem:mq1}
$ X \mr S \by T \Rightarrow X \parallel T \mr S $
\end{lemma}
\begin{proof}
Assume that $X \mr S \by T$.
We let
\[ \mathsf R = \{ (x \parallel t, s) \mid x \mr \{ s \by t \} \} \]
and show that $\mathsf R$ is a~modal refinement relation.
Let $(x \parallel t, s) \in \mathsf R$.

\begin{itemize}
\item Let $x \parallel t \may{a} x' \parallel t'$. As $x \mr \{s \by t\}$
 this means that $\{ s \by t \} \may{a} \{s_1\by t_1,\ldots,\linebreak[4] s_k\by t_k\}$
 and $x' \mr \{s_1\by t_1,\ldots,s_k\by t_k\}$. 
 Due to the construction of $\May_a(\{s \by t\})$), we know
 that one of the $t_j = t'$ and $s \may{a} s_j$. Let $s' = s_j$.
 Due to Lemma~\ref{lem:mq0}, $x' \mr \{s' \by t'\}$. Therefore,
 $(x'\parallel  t',s') \in \mathsf R$.
\item Let $s \must{a} s'$. This means that $\{ s \by t \} \must{} U$.
 As $x \mr \{ s \by t \}$, we know that $x \must{a} x'$ and $x' \mr u$ 
 where $u \in U$. Due to construction of $U$ we know that there exists
 $s' \by t' \in u$. Again, due to Lemma~\ref{lem:mq0}, $x' \mr u \mr \{s'\by t'\}$.
 Therefore, $(x'\parallel  t',s') \in \mathsf R$.
\end{itemize}
Clearly, $x_0 \mr s_0 \by t_0$ and thus $(x_0 \parallel t_0,s_0) \in \mathsf R$
which means that $X \parallel T \mr S$.
\end{proof}

\begin{lemma}\label{lem:mq2}
$ X \parallel T \mr S \Rightarrow X \mr S \by T $
\end{lemma}

\begin{proof}
Assume that $X \parallel T \mr S$.
We let
\[ \mathsf R = \{ (x, \{ s_1 \by t_1, \ldots, s_k \by t_k \}
	\mid \forall j : x \parallel t_j \mr s_j \} \]
and show that $\mathsf R$ is a~modal refinement relation.

Let now $(x,q =  \{ s_1 \by t_1, \ldots, s_k \by t_k \})
\in \mathsf R$.

\begin{itemize}
\item Let $x \may{a} x'$. Take an arbitrary $t_i \may{a} t_{i,j}$.
	We have $x \parallel t_i \may{a} x' \parallel t_{i,j}$ and
	as $x \parallel t_i \mr s_i$ we also have a~corresponding
	$s_i \may{a} s_{i,j}$ with $s_{i,j} \mr x' \parallel t_{i,j}$.
	We fix these $s_{i,j}$.
	Let $q' = \{ s_{i,j} \by t_{i,j} \mid i \in \{1,\ldots,k\},
	 j \in \{1,\ldots,m_i\}\}$.
	Clearly, $q \may{a} q'$ and $(x',q') \in \mathsf R$. 
\item Let $q \must{} U$ and let $s_j \must{a} s'_j$ be the corresponding
	must transition in the construction. 
	As $x \parallel t_j \mr s_j$, this means that $x \must{a} x'$ and
	$t_j \must{a} t'_j$ such that $x' \parallel t'_j \mr s'_j$.
	This also means that $x \may{a} x'$. We thus build $q'$ as we did
	in the previous case. Clearly, $t'_j = t_{j,h}$ for some $h$.
	Let $\bar q = \{ \bar s \by \bar t \in q \mid \bar t \ne t_{j,h} \}
	\cup \{ s'_j \by t'_j \}$.
	Due to the construction of must, $\bar q \in U$.
	Clearly $(x',\bar q) \in \mathsf R$.
\end{itemize}
We know that $x_0 \parallel t_0 \mr s_0$. Thus also $(x,\{ s_0\by t_0\})
\in \mathsf R$ which means that $X \mr S \by T$.
\end{proof}

\section{Pruning}\label{app:pruning}

%\subsection{Consistency and Pruning}

For practical application of our translations, and also for some of the
constructions we present in the paper, it can be beneficial to reduce
specifications to their part which is \emph{reachable} and
\emph{consistent}.  As an example, a DMTS state $s$ with $s\must{}
\emptyset$ will admit no implementation and can be removed, but then all
transitions leading to it must also be removed.  This is the intuition
of our pruning constructions which we give for DMTS and NAA, and which
are based on the construction for MTS introduced
in~\cite{DBLP:journals/mscs/BauerJLLS12}. 

The set of \emph{reachable} states $\Reach( S)$ in a NAA $( S, S^0,
\Tran)$ is defined as usual, by declaring that $S^0\subseteq \Reach( S)$
and, recursively, for all $s\in \Reach( S)$, all $M\in \Tran( s)$ and
all $( a, t)\in M$, that $t\in \Reach( S)$.  We say that a state $s\in
S$ is \emph{locally consistent} if $\Tran( s)\ne \emptyset$, and that
$S$ itself is locally consistent if $\Reach( S)\ne \emptyset$ and all
$s\in \Reach( S)$ are locally consistent.

\begin{lemma}
  \label{le:lcons->cons}
  For any locally consistent NAA $S$, $\impl S\ne \emptyset$.
\end{lemma}
\begin{proof}%[Proof of Lemma~\ref{le:lcons->cons}]
  Let $S'= \Reach( S)$, and choose for each $s\in S'$, arbitrarily,
  precisely one $M\in \Tran( s)$ and define $\Tran'( s)=\{ M\}$.  The
  so-defined NAA $( S', s^0, \Tran')$ is an implementation with $S'\mr
  S$.
\end{proof}

The following \emph{pruning} algorithm may be used to turn consistent
NAA into locally consistent ones: For a given NAA $( S, S^0, \Tran)$,
define the predecessor mapping $\pred: 2^S\to 2^S$ by $\pred( B)=\{ s\in
S\mid \forall M\in \Tran( s): \exists( a, t)\in M: t\in B\}$.  Denote
by $\pred^*$ the reflexive, transitive closure of $\pred$, and let $B'=
\pred^*(\{ s\in S\mid \Tran( s)= \emptyset\})$.  The pruning of $S$
is defined to be $\rho( S)= ( S\setminus B', S^0\setminus B', \Tran')$,
with $\Tran'( s)=\{ M\models \Tran( s)\mid \forall( a, t)\in M: t\in
S\setminus B'\}$.

\begin{lemma}
  \label{le:pruning}
  For any NAA $S$ and any locally consistent NAA $T$, $T\mr S$ iff $T\mr
  \rho( S)$.
  % \begin{enumerate}
  % \item\label{le:rho.cons} $\rho( S)\not\mreq \bot$ iff $S$ is
  %   consistent,
  % \item\label{le:rho.lcons} if $\rho( S)\not\mreq \bot$, then $\rho( S)$
  %   is locally consistent,
  % \item\label{le:rho.impl} $\impl{ \rho( S)}= \impl{ S}$.
  % \end{enumerate}
\end{lemma}
\begin{proof}%[Proof of Lemma~\ref{le:pruning}]
  Denote $S=( S, S^0, \Tran)$ and $\rho( S)=( S', S^0, \Tran')$.  The
  backward direction is clear in case $\rho( S)\mreq \bot$, so let $T=(
  T, T^0, \Tran_T)$.  Let $R\subseteq T\times S'$ be a modal refinement
  witnessing $T\mr \rho( S)$, then $R\subseteq T\times S$ is easily seen
  to be a witness for $T\mr S$.

  For the forward direction, assume again first that $\rho( S)\mreq
  \bot$.  By construction of $\rho( S)$, we know that for any $s^0\in
  S^0$ there exists a sequence $( s_1,\dots, s_n)$ of states in $S$ such
  that $s_1= s^0$, for all $j= 1,\dots, n- 1$ and for all $M\in \Tran(
  s_j)$, there is some $( a_{ j+ 1}, s_{ j+ 1})\in M$, and $\Tran( s_n)=
  \emptyset$.  Now assume that there is a NAA $T\mr S$, then by
  refinement, $T$ must contain a similar sequence $( t_1,\dots, t_n)$ of
  states, with $t_1\in T^0$, such that for all $j= 1,\dots, n- 1$, there
  is $M\in \Tran_T( t_j)$ with some $( a_{ j+ 1}, t_{ j+ 1})\in M$.  But
  then $\Tran_T( t_n)= \emptyset$, so that $T$ is not locally
  consistent.

  Now let $T$ be a NAA with $T\mr S$ and $R\subseteq T\times S$ a
  witness.  If there is $( t, s)\in R$ with $s\in S\setminus S'$, then
  by the same argument as above, $T$ is locally inconsistent.  Hence
  $R\subseteq T\times S'$ is a witness for $T\mr \rho( S)$.
%
  % To prove point~\ref{le:rho.impl}, we only have to note that also for
  % all implementations $I$, $I\mr S$ iff $I\mr \rho( S)$.
  % Point~\ref{le:rho.lcons} is clear as all locally inconsistent states
  % in $S$ are removed during construction of $\rho( S)$.  For
  % point~\ref{le:rho.cons}, if $S$ is consistent, then there is an
  % implementation $I\mr S$, hence also $I\mr \rho( S)$, implying $\rho(
  % S)\not\mreq \bot$.  For the other direction, if $\rho( S)\not\mreq
  % \bot$, then $\rho( S)$ is locally consistent, hence consistent by
  % Lemma~\ref{le:lcons->cons}, so that there is an implementation $I\mr
  % \rho( S)$, but then also $I\mr S$.
\end{proof}

As a consequence, $\impl{ \rho( S)}= \impl{ S}$ for all NAA $S$.
% Also, $\rho( S)$ is either inconsistent or locally consistent, and
% consistent iff $S$ is consistent.
%
We also introduce pruning for DMTS.  For a DMTS $( S, S^0, \omay,
\omust)$, define $\pred_D: 2^S\to 2^S$ by $\pred_D( B)=\{ s\in S\mid
\exists s\must{} N: \forall( a, t)\in N: t\in B\}$.  Let $B'=\pred_D^*\{
s\in S\mid s\must{} \emptyset\}$ and define $\rho_D( S)= ( S\setminus
B', S^0\setminus B', \omay', \omust')$, with $\omay'= \omay\cap(
S\setminus B')\times \Sigma\times( S\setminus B')$ and $\omust'=\{( s,
N')\in S\times 2^{ \Sigma\times S}\mid s\in S\setminus B', \exists( s,
N)\in \omust: N'= N\cap \Sigma\times( S\setminus B')\}$.
  
\begin{lemma}
    \label{le:predD=pred}
    For any DMTS $S$ and any $B\subseteq S$, $\pred( B)= \pred_D( B)$.
  \end{lemma}

  \begin{proof}%[Proof of Lemma~\ref{le:predD=pred}]
    Let $s\in S$.  We have $\Tran( s)=
%% \{ M\subseteq \Sigma\times S\mid M\models \Phi( s)\}=
  \{ M\subseteq \Sigma\times S\mid \forall
    s\must{} N: M\cap N\ne \emptyset, \{ s\}\times M\subseteq \omay\}$.
    Hence $s\in \pred( B)\liff \forall M\subseteq \Sigma\times S:\{
    s\}\times M\not\subseteq \omay\lor M\cap \Sigma\times B\ne
    \emptyset\lor \exists s\must{} N: M\cap N= \emptyset$.

    Now assume $s\in \pred_D( B)$, then we have $s\must{} N$ for which
    $N\subseteq \Sigma\times B$.  Let $M\subseteq \Sigma\times S$.  If
    $M\cap \Sigma\times B\ne \emptyset$ or $\{ s\}\times M\not\subseteq
    \omay$, we are done.  If $\{ s\}\times M\subseteq \omay$ and $M\cap
    \Sigma\times B= \emptyset$, then $N\subseteq \Sigma\times B$ implies
    that also $M\cap N= \emptyset$.  We have shown that $s\in \pred(
    B)$.

    Assume $s\notin \pred_D( B)$, then it holds for all $s\must{} N$
    that there is $( a, t)\in N$ with $t\notin B$.  Define $M=\{( a,
    t)\in \Sigma\times S\mid \exists s\must{} N:( a, t)\in N, t\notin
    B\}$.  Then $\{ s\}\times M\subseteq \omay$ and $M\cap \Sigma\times
    B= \emptyset$.  Now let $s\must{} N$, then we have $( a, t)\in N$
    for which $t\notin B$.  But then also $( a, t)\in M$, hence $M\cap
    N\ne \emptyset$.
  \end{proof}

\begin{lemma}
  \label{le:pruning-same}
  For all DMTS $S$, $\rho_D( S)= \rho( S)$.
\end{lemma}
\begin{proof}%[Proof of Lemma~\ref{le:pruning-same}]
%%  First a technical lemma:
%%  Now for the proof of Lemma~\ref{le:pruning-same}: 
  In light of
  Lemma~\ref{le:predD=pred}, it suffices to show that for all $s\in S$,
  $s\must{} \emptyset$ iff $\Tran( s)= \emptyset$.
%% , \ie~iff $\Phi(  s)\equiv \ff$.  
   Now if $s\must{} \emptyset$, then indeed $\Tran(s) = \emptyset$
   by definition of $\Tran$.
%% $\Phi( s)=   \ff$ by definition of $\Phi$.

  For the other direction, assume $\Tran( s)=\{ M\subseteq \Sigma\times
  S\mid \forall s\must{} N: M\cap N\ne \emptyset, \{ s\}\times
  M\subseteq \omay\}= \emptyset$.  Then for all $M\subseteq \Sigma\times
  S$ with $\{ s\}\times M\subseteq \omay$, we must have $s\must{} N$
  with $M\cap N= \emptyset$.

  Now let $M=\{( a, t)\mid( s, a, t)\in \omay\}$, so that we have
  $s\must{} N$ with $M\cap N= \emptyset$.  Assume that there is $( a,
  t)\in N$, then also $( s, a, t)\in \omay$ and hence $( a, t)\in M$, so
  that $M\cap N\ne \emptyset$, a contradiction.  Thus we must have $N=
  \emptyset$.
\end{proof}\bigskip

\end{document}